\documentclass[10pt,twocolumn]{IEEEtran}

\usepackage[utf8]{inputenc}
\usepackage{amssymb,amsmath,amsfonts,amsthm,bm}
\usepackage{mathrsfs}
\usepackage{graphicx}
\usepackage{cite}
\usepackage{color}
\usepackage{lipsum}
\usepackage[bookmarks]{hyperref}
\usepackage{booktabs}
\usepackage[mathscr]{euscript}
\usepackage{soul}
\usepackage[scr]{rsfso}
\usepackage{tablefootnote}
\usepackage{enumerate}
\usepackage{array,threeparttable}
\usepackage{makecell}
\usepackage{relsize}
\usepackage{bbm}

\newcolumntype{M}[1]{>{\centering\arraybackslash}m{#1}}
\def\BibTeX{{\rm B\kern-.05em{\sc i\kern-.025em b}\kern-.08em
    T\kern-.1667em\lower.7ex\hbox{E}\kern-.125emX}}

\usepackage[font=footnotesize]{caption}
\usepackage{float}
\usepackage[list=true]{subcaption}
\usepackage{multirow}
\usepackage[export]{adjustbox}
\def\BibTeX{{\rm B\kern-.05em{\sc i\kern-.025em b}\kern-.08em
    T\kern-.1667em\lower.7ex\hbox{E}\kern-.125emX}}
\makeatletter
  \def\hrulefill{\leavevmode\leaders\hrule height 1pt\hfill\kern\z@}
\makeatother
\usepackage{wrapfig}
\usepackage{listings}

\definecolor{gris245}{RGB}{245,245,245}
\definecolor{olive}{RGB}{50,140,50}
\definecolor{brun}{RGB}{175,100,80}
\usepackage{fancybox}

\hypersetup{
  unicode=false,
  colorlinks=true,       
  linkcolor=black,       
  citecolor=black,       
  filecolor=black,       
  urlcolor=black,
  bookmarksdepth=4
}

\newcommand{\Alice}[1]{\textnormal{\texttt{A}}}
\newcommand{\Bob}[1]{\textnormal{\texttt{B}}}
\newcommand{\Eve}[1]{\textnormal{\texttt{E}}}

\newtheorem{proposition}{Proposition}
\newtheorem{corollary}{Corollary}

\newtheorem*{definition}{Definition}

\begin{document}

\title{Partial Secrecy Analysis in Wireless Systems: Diversity-Enhanced PLS over Generalized \\Fading Channels}

\author{
Henry~Carvajal~Mora, \textit{Senior Member, IEEE},  Nathaly Orozco, \textit{Senior Member, IEEE}, \\ Fernando Almeida García, \textit{Senior Member, IEEE}, José Vega-Sánchez, \textit{Senior Member, IEEE}, \\ Felipe Grijalva, \textit{Senior Member, IEEE}, and Edgar Benitez Olivo, \textit{Member, IEEE}
\thanks{H.~C.~Mora, J. V.-Sánchez, and F.~Grijalva are with the Colegio de Ciencias e Ingenierías "El Politécnico", Universidad San Francisco de Quito (USFQ), Diego de Robles S/N, Quito 170157, Ecuador (e-mails: \{hcarvajal,dvega,fgrijalva\}@usfq.edu.ec).

N.~Orozco is with the Networking and Telecommunications Engineering, ETEL Research Group, Faculty of Engineering and Applied Sciences, Universidad de Las Américas (UDLA), 170503 Quito, Ecuador (e-mail: nathaly.orozco@udla.edu.ec).

F.~D.~A.~García is with the Wireless and Artificial Intelligence (WAI) laboratory, National Institute of Telecommunications (INATEL), Santa Rita do Sapucaí, MG, 37540-000, Brazil (e-mail: fernando.almeida@inatel.br).

E.~B.~Olivo is with the Department of Communications, School of Electrical and Computer Engineering, University of Campinas (UNICAMP), Campinas 138083-852, São Paulo, Brazil (e-mail: ebenitez@unicamp.br).}
\thanks{
This work was supported in part by the Universidad de Las Américas (UDLA), Quito, Ecuador under Project 577.A.XVII.25, by the Universidad San Francisco de Quito through the Poli-Grants Program under Grant 33595, in part by the National Council for Scientific and Technological Development (CNPq) under Grant 420148/2023-0, and in part by the Fundo de Apoio ao Ensino, Pesquisa e Extensão da UNICAMP (FAEPEX) under Grants 2198/25 and 2621/25.
This work was also funded by the Brasil 6G Project with support from RNP/MCTI (Grant 01245.010604/2020-14), and by the xGMobile Project (Code \mbox{XGM-AFCCT-2025-8-1-1}) with resources from EMBRAPII/MCTI (Grant 052/2023 PPI IoT/Manufatura 4.0 and FAPEMIG Grant \mbox{PPE-00124-23})
}
}

\maketitle

\begin{abstract}
Securing information in future mobile networks is particularly challenging for devices with limited computational resources. In such contexts, physical layer security (PLS) offers a promising alternative by exploiting the inherent randomness of wireless channels. When full secrecy cannot be guaranteed, the partial secrecy regime provides a realistic and effective approach to system design.
This study investigates the partial secrecy performance of wireless systems operating over the recently introduced generalized multicluster fluctuating two-ray (MFTR) fading model. This model is highly regarded for its adaptability in representing a broad range of propagation environments and includes various classical fading models as exceptional cases.
We consider a secure communication setup involving a transmitter (\Alice{}), an intended receiver (\Bob{}), and an eavesdropper (\Eve{}). Inspired by previous work showing the benefits and simplicity of diversity schemes for improving PLS, we assume that both \Bob{} and \Eve{} are equipped with antenna arrays and apply maximal ratio combining (MRC). The model also considers independently but not identically distributed (i.n.i.d.) fading conditions on the $\Alice{}\rightarrow\Bob{}$ and $\Alice{}\rightarrow\Eve{}$ links.
We derive exact and closed-form approximate expressions for three central metrics in partial secrecy theory: generalized secrecy outage probability (GSOP), average fractional equivocation (AFE), and average information leakage rate (AILR). Unlike previous approaches based on simpler fading models, the proposed expressions maintain a constant computational complexity, irrespective of the number of diversity branches, and are compatible with standard mathematical tools.
Validation via Monte Carlo simulations confirms the accuracy of the derived expressions and highlights the impact of key system parameters on secrecy performance. The flexible MFTR framework enables assessment across varied fading conditions. In particular, increasing the number of MRC branches at \Bob{} improves the GSOP diversity order and alters PLS behavior based on the fading characteristics of the $\Alice{}\rightarrow\Eve{}$ link.
\end{abstract}


\section{Introduction}
\label{sec:introduction}

\IEEEPARstart{A}{s} wireless communications evolve toward the sixth generation (6G), ensuring resilient and efficient information security is more critical than ever. The vision of 6G includes ultra-dense deployments, ubiquitous connectivity, and massive machine-type communications (mMTC), all of which dramatically increase the number of potential attack vectors. This growing complexity amplifies the risk of eavesdropping and data breaches, particularly in scenarios involving resource-constrained devices. As a result, there is a pressing need for security solutions that are not only robust but also lightweight and scalable.

Physical layer security (PLS) is a complementary approach to traditional cryptographic methods in telecommunications, offering resilience against eavesdropping without relying on high computational complexity. This is particularly important in scenarios where conventional security mechanisms may be impractical due to latency constraints, energy limitations, or large-scale deployments. Moreover, the adoption of millimeter-wave (mmWave), sub-terahertz (sub-THz), and terahertz (THz) frequencies, expected to be fundamental in 6G, introduces unique propagation characteristics that influence security analysis and necessitate novel countermeasures tailored to these high-frequency environments.

Classical PLS assumes perfect secrecy, often unrealistic in practical wireless systems. Moreover, the classical secrecy outage probability (SOP) definition is restrictive, as it only indicates whether secrecy is lost without quantifying the eavesdropper’s decoding capability or the amount of leaked information. To address this, the partial secrecy regime introduces refined metrics that provide a more practical assessment of confidentiality \cite{Biao_2016}. 

Fractional equivocation (FE) is a fundamental security metric for analyzing the partial secrecy regime. FE quantifies the degree of uncertainty at the eavesdropper, reflecting how much of the transmitted information remains indecipherable \cite{Moya_2020}. Building on this concept, alternative secrecy metrics have been developed to evaluate security comprehensively in quasi-static fading channels. These metrics include the generalized secrecy outage probability (GSOP), which quantifies the probability of secrecy being compromised at different levels of confidentiality; the average fractional equivocation (AFE), which measures the average degree of uncertainty at the eavesdropper; and the average information leakage rate (AILR), which assesses the average amount of information leaked over time \cite{Biao_2016}. These metrics offer deeper insights into the security performance of wireless systems and are formally defined in subsequent sections.

The security of wireless systems has been extensively studied under various scenarios, often relying on traditional fading models such as Rayleigh, Rician, Nakagami-$m$, and $\kappa$-$\mu$ \cite{Mahmoud_2017,Yun_2018,Akter_2019,Devi2023,Wei_2024,Mora_2024}. While effective in certain propagation conditions, these models struggle to accurately capture the complexities of next-generation wireless networks, particularly those operating in mmWave, sub-THz, and THz frequency ranges \cite{Marins21,Tekbıyık_2021,Papasotiriou_2021}. Higher-frequency bands introduce unique propagation challenges, requiring more adaptable fading models that better reflect real-world conditions.

A step in this direction is embodied by the multi-cluster fluctuating two-ray (MFTR) fading model, introduced in \cite{Vega_2023}, which extends the characteristics of the fluctuating two-ray (FTR) and $\kappa$-$\mu$ shadowed fading models \cite{Paris_2014}. The MFTR model offers versatility in characterizing complex fading environments by incorporating additional multipath clusters and fluctuating specular components. Its probability density function (PDF) can exhibit bimodal behavior under specific parameter configurations, making it particularly well-suited for modeling fading in high-frequency bands~\cite{Vega_2023}. Furthermore, the MFTR distribution generalizes several well-known models, including one-sided Gaussian, Rayleigh, Nakagami-$q$ (Hoyt), Nakagami-$m$, Rician, Rician shadowed, $\kappa$-$\mu$, $\kappa$-$\mu$ shadowed, $\eta$-$\mu$, two-wave, two-wave with diffuse power (TWDP), and FTR.

Additionally, diversity techniques provide a simple yet effective means of enhancing security in wireless systems. By combining multiple independent signal paths, diversity improves the reliability of communication links for legitimate users, ensuring robust performance even in adverse propagation conditions. This enhanced reliability directly strengthens PLS by reducing the probability of successful interception by an eavesdropper \cite{He_2011,Mahmoud_2017,Akter_2019,Rizve_2024,Hasan_2024,Rizve_2024_2,Donghun_2024}.

\subsection{Literature Review}

Some studies have examined PLS in wireless systems operating over channels modeled with generalized fading distributions, incorporating diversity techniques to enhance performance and security. These studies explore advanced statistical models to capture the diverse and complex propagation conditions encountered in modern and future wireless communication scenarios. For instance, in \cite{Wang_2014}, the PLS of maximal ratio combining (MRC) systems over TWDP fading \cite{Durgin_2002} is analyzed considering that a multiple-antenna eavesdropper intercepts messages sent from a single antenna transmitter to a multiple-antenna receiver. An exact integral-form expression for the average secrecy capacity (ASC) and an expression formulated as nested infinite summations for the SOP are derived. Considering the high signal-to-noise ratio (SNR) regime, asymptotic expressions are also provided. 

Other works explore PLS considering the FTR fading distribution introduced in \cite{ftr}. For example, in \cite{Ouyang_2020}, a scenario is assumed where the legitimate receiver and eavesdropper operate with a single antenna under FTR fading, while the transmitter employs frequency-diverse arrays (FDA). That study provides closed-form approximate expressions for the average secrecy rate (ASR) and SOP. These expressions involve more than two nested infinity summations. 
A similar scenario is analyzed in \cite{Zhao_2019}, but it considers the alternative definition of the SOP described in \cite{Zhou_2011}, which states that when the received SNR at the legitimate receiver surpasses a specific threshold, the source sends signals to this receiver.

In \cite{Cheng_2021}, a multi-carrier FDA model is introduced, considering range, azimuth angle, and elevation angle for PLS transmission over generalized FTR fading channels. The study presents integral-form expressions involving nested infinite summations to calculate lower and upper bounds for the SOP. 
In \cite{Gao_2022}, a multiple-input single-output (MISO) wiretap channel model is examined, featuring a transmitter equipped with $N$ antennas and a receiver and an eavesdropper, each equipped with a single antenna. 
The transmitter selects the transmit antenna with the highest instantaneous SNR. Expressions for ASC and SOP are derived, with the former presented in integral form and the latter expressed as nested infinite summations, where the number of summations is proportional to $N$. 

In~\cite{Mathur_2018}, the secrecy performance of wireless systems over $\alpha$--$\eta$--$\kappa$--$\mu$ fading channels is analyzed for the classic Wyner’s wiretap scenario, where the transmitter, the legitimate receiver, and the eavesdropper are all assumed to be single-antenna nodes. The authors derive closed-form expressions for the ASC and SOP, which are expressed in terms of advanced special functions, including the multivariate Fox H and Meijer G functions. While the results are mathematically rigorous, it is worth noting that these special functions are not built-in features of widely used computational tools such as MATLAB or Wolfram Mathematica.

Recent efforts have explored secrecy performance under correlated fading conditions. In particular, \cite{Mathur_2019} assumes that both the main and wiretap links experience correlated $\alpha$--$\mu$ fading. The authors derive novel exact expressions for the ASC and SOP, taking into account the effect of channel correlation. These expressions, however, are also obtained in terms of the multivariate Fox H and Meijer G functions.

In~\cite{Kong_2019}, the authors analyze PLS over wiretap channels whose fading distribution is modeled with the general Fox H-function. Expressions for secrecy metrics, such as SOP and ASC, are derived in both non-colluding and colluding eavesdropper scenarios, with results expressed in terms of univariate, bivariate, or multivariate Fox H-functions.

All the previously discussed works rely on the classical information-theoretic definition of secrecy, which assumes that the eavesdropper experiences complete decoding failure, i.e., essentially adhering to the classical SOP criterion. 

Under the partial secrecy regime, \cite{Sharma_2023} presents a secrecy performance analysis of a system consisting of a transmitter, a receiver, and an eavesdropper. In this system model, all nodes are single-antenna devices and operate over an FTR fading channel. An exact expression for calculating the GSOP was obtained in terms of an infinite summation. Furthermore, exact expressions for calculating the AFE and AILR were derived and written in terms of four infinite summations. Additionally, asymptotic expressions for these metrics are obtained, considering that the received SNR at the eavesdropper tends to infinity. Unfortunately, the numerical results do not include these asymptotes, preventing validation of their accuracy across all scenarios.

The numerical results reported in \cite{Zhao_2019} and \cite{Sharma_2023} assume independent and identically distributed (i.i.d.) fading for the links between the transmitter and both the legitimate receiver and the eavesdropper. However, this assumption may not hold in practical systems due to differences in propagation environments among the links \cite{Gao_2022}.

The authors in \cite{Sharma_2023} also present an analysis of a multi-branch system in which only the eavesdropper uses an antenna array. Therein, the cumulative distribution function (CDF) of the eavesdropper's SNR is derived. However, a limitation of their approach is that the resulting CDF contains nested infinite summations that scale proportionally with the number of antennas in the eavesdropper’s array.
This complexity hinders theoretical analysis and adds a significant computational burden when evaluating secrecy metrics, such as GSOP, AILR, and AFE.

\subsection{Motivation and Contributions}

Drawing from the literature review and to the best of the authors' knowledge, existing secrecy performance analyses frequently yield intricate mathematical formulations, where the computational burden increases significantly as the number of receiving antennas at the eavesdropper increases. 
Moreover, some existing studies consider scenarios in which only the eavesdropper is equipped with an antenna array, often to model a worst-case secrecy condition. However, in practical deployments, the legitimate receiver is also likely to be equipped with multiple antennas and can benefit from diversity techniques such as MRC. Accounting for this capability is important, as it not only enhances the overall secrecy performance but also leads to a more general and realistic system model, particularly relevant for modern wireless networks, where multi-antenna receivers are becoming increasingly common.



Traditional secrecy metrics such as ASC and SOP follow a binary model of confidentiality, assuming either perfect secrecy or total leakage. However, this assumption often fails under quasi-static fading, where the eavesdropper may still extract partial information due to finite blocklengths, coding imperfections, or favorable channel realizations~\cite{Biao_2016}. To address these limitations, the partial secrecy framework introduces more refined metrics---such as FE, GSOP, AFE, and AILR---which offer a more realistic quantification of the eavesdropper’s uncertainty and the information exposure. Unlike SOP, which only measures the probability of outage, the aforementioned metrics provide deeper insights into how much information is leaked, making them more suitable for the analysis and design of future wireless systems.



While previous studies have analyzed secrecy performance under classical fading distributions such as Rayleigh, Nakagami-$m$, Rician, or $\kappa$--$\mu$ shadowed models, these fading models represent particular cases of more general frameworks. Recent works have adopted generalized fading models whose distributions are expressed via the Fox H-function, a versatile mathematical construct capable of encapsulating a wide range of fading distributions within a single formulation. However, the resulting expressions for secrecy metrics often involve univariate or multivariate Fox H-functions, whose numerical evaluation requires complex contour integration and specialized mathematical tools, making them not trivial to implement using standard mathematical software. In contrast, the MFTR fading model is a recently proposed physical model that offers a balance between modeling versatility and computational practicality. Specifically, it generalizes the $\kappa$--$\mu$ shadowed and FTR models by integrating clustered multipath with jointly fluctuating specular components—features that are particularly relevant in mmWave, sub-THz, and THz propagation. Despite its expressive power and physical grounding, the MFTR model has not yet been explored in the context of PLS, particularly under partial secrecy metrics, which reveals a meaningful gap in the literature.


Finally, although the fading affecting the legitimate and wiretap links is often considered statistically independent due to the spatial separation between the legitimate receiver and the eavesdropper \cite{Sharma_2023}, this assumption does not preclude the possibility that the fading statistics for these links may be either identical or different. In covert eavesdropping scenarios, the eavesdropper may be located near the legitimate receiver and experience similar propagation conditions, leading to identically distributed fading \cite{Mathematics2022}. Conversely, in many practical deployments, environmental asymmetries may cause the two links to exhibit distinct fading characteristics. Hence, it is important to adopt a flexible framework that supports both i.i.d. and independent non-identically distributed (i.n.i.d.) fading, enabling the evaluation of secrecy performance across diverse conditions.

By tackling these critical aspects—the necessity of simplified secrecy expressions in diversity-based systems, the relevance of the partial secrecy regime, the importance of adopting a general fading model, and the need to account for i.n.i.d fading conditions—our work aims to bridge the existing research gap and provide deeper insights into the PLS of future wireless networks.

This work considers a secrecy communication system comprising a transmitting source (\Alice{}), a legitimate receiver (\Bob{}), and an external eavesdropper (\Eve{}). Unlike previous studies that focus on scenarios where devices either use a single antenna or only \Eve{} employs an antenna array, our system model considers a more general case: both \Bob{} and \Eve{} are equipped with antenna arrays and apply MRC to exploit diversity in signals experiencing i.i.d. fading 
across branches at each receiver, but i.n.i.d. fading between the legitimate and wiretap links. The secrecy performance is evaluated in the context of the partial secrecy regime, focusing on systems operating over the generalized MFTR fading model. 
In addition to the points above, one of the objectives of this work is not only to introduce a new comprehensive channel model in secrecy analysis but also to derive mathematically tractable expressions for secrecy performance metrics, including GSOP, AFE, and AILR. A key additional goal is to ensure that these expressions are computationally efficient and maintain low complexity, regardless of the number of diversity branches employed by either \Bob{} and \Eve{}. In this sense, the main contributions of this work can be summarized as follows:
\begin{itemize}
    
    \item Using the FE framework, we derive exact expressions for computing the GSOP, AFE, and AILR. These expressions account for the possibility that the MFTR fading channels affecting the $\Alice{}\rightarrow\Bob{}$ and $\Alice{}\rightarrow\Eve{}$ links are identically or non-identically distributed.
    
    \item Unlike prior studies on simpler fading models, our secrecy performance expressions preserve a consistent computational complexity, irrespective of the number of diversity branches at both \Bob{} and \Eve{}. In addition, these expressions can be easily computed using standard mathematical software such as Matlab or Mathematica.
    
    \item Accurate closed-form approximate expressions are derived for calculating the GSOP, AFE, and AILR in the high SNR regime. Also, an asymptotic expression is obtained for the GSOP, revealing that the diversity gain is determined by the product of the number of diversity branches at \Bob{} and the total number of multipath clusters within the MFTR fading model of the $\Alice{}\rightarrow\Bob{}$ link. The asymptote further reveals that the coding gain depends on the fading distribution parameters of the $\Alice{}\rightarrow \Bob{}$ and $\Alice{}\rightarrow \Eve{}$ links, the number of diversity branches at both \Bob{} and \Eve{}, and the SNR at \Eve{}.
  
    \item Monte Carlo simulations confirm the accuracy of our expressions and show how system parameters impact secrecy performance. Leveraging the flexibility of the MFTR fading model, we analyze system security under various fading channel conditions, including specific cases of this model for both the $\Alice{}\rightarrow\Bob{}$ and $\Alice{}\rightarrow\Eve{}$ links. We provide a graphical analysis of the GSOP, including the CDF of the received SNR at \Bob{} and the PDF of the received SNR at \Eve{}, offering a detailed interpretation of the results. Our findings highlight how PLS is highly dependent on the fading conditions affecting both $\Alice{}\rightarrow\Bob{}$ and $\Alice{}\rightarrow\Eve{}$ links.

\end{itemize}

\subsection{Paper Organization and Notation}

The remainder of this manuscript is organized as follows. In Section \ref{sec:System_Model}, we introduce the system and channel models. Subsequently, Section \ref{sec:Exact_Analysis} conducts the secrecy analysis, yielding exact expressions to assess the GSOP, AFE, and AILR. Furthermore, in Section \ref{sec:Approximate_Analysis}, we provide approximate and asymptotic expressions to compute these metrics under high SNR conditions. Section \ref{sec:Num_Results} encompasses numerical results and extensive discussions of the findings. Finally, Section \ref{sec: Conclusions} closes the work with concluding remarks.

\textit{Notation:} $\mathbb{E} \left[\cdot\right]$ denotes expectation, $\Gamma(\cdot)$ is the gamma function \cite[eq. (6.1.1)]{abramowitz72}, $\Upsilon(\cdot)$ is the lower incomplete gamma function \cite[eq. (6.5.2)]{abramowitz72}, $\mathcal{B}(\cdot,\cdot,\cdot)$ is the incomplete Beta function \cite[eq. (1)]{Osborn_1968}, $\log(\cdot)$ is the natural logarithm, $(\cdot)_{(\cdot)}$ is the Pochhammer symbol \cite[eq. (5.2.3)]{Olver10}, $_1 F_1(\cdot,\cdot,\cdot)$ is the Kummer's confluent hypergeometric function \cite[eq. (13.1.2)]{abramowitz72}, $_2 F_1(\cdot,\cdot;\cdot;\cdot)$ is the Gauss hypergeometric function \cite[eq. (15.1.1)]{abramowitz72}, $\mathbb{Z}$ and $\mathbb{C}$ are the sets of integer and complex numbers, respectively, $\mathcal{U}[x,y]$ is a uniformly distributed random variable (RV) between $x$ and $y$, $\mathcal{CN}(0,x)$ is a zero-mean complex normal RV with variance $x$, and $\mathbbm{i}=\sqrt{-1}$.

\section{System and Channel Models}
\label{sec:System_Model}

Consider the wireless system in Fig. \ref{fig:SystemModel}, where \Alice{} sends uncoded information to \Bob{}, which combines the signals captured from $L_{\Bob{}}$ diversity branches (antennas) by using MRC. An external eavesdropper, \Eve{}, is also present, attempting to intercept and decode the information transmitted by \Alice{} through the capture of signals via $L_{\Eve{}}$ diversity branches.
Likewise, \Eve{} employs MRC to meet its goal. It is assumed that both \Bob{} and \Eve{} have awareness of the channel state information (CSI) associated with the links originating from \Alice{}.

\begin{figure}[t]
    \centerline{\includegraphics[width=\linewidth]
    {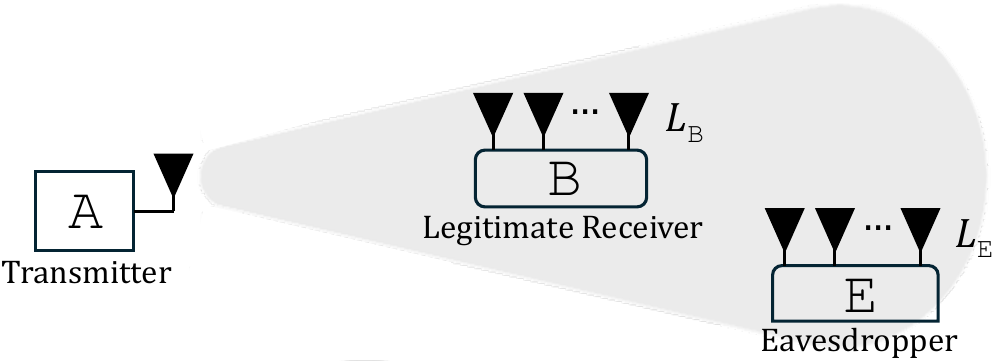}}
    \caption{System model.}
    \label{fig:SystemModel}
    \vspace{-3mm}
\end{figure}

The signals captured by \Bob{} and \Eve{} on their respective diversity branches experience i.i.d. MFTR fading, while exhibiting i.n.i.d. fading between the legitimate and wiretap links. Hence, the instantaneous SNRs at the outputs of the MRC 
receivers at \Bob{} and \Eve{} can be written as \cite[eq. (19)]{Da_Costa_2008}
\begin{align} \label{eq:Total_SNR_Bob}
    \Psi_{\texttt{X}} = \sum_{\ell=1}^{L_{\texttt{X}}} \underbrace{\frac{\overline{\gamma_{\texttt{X}}}\left| \alpha_{\ell,\texttt{X}}\right|^2}{\Omega_{\ell,\texttt{X}}}}_{\gamma_{\ell,\texttt{X}}},
\end{align}
for $\texttt{X}\in\{\Bob{},\Eve{}\}$, where $\overline{\gamma_{\Bob{}}}$ and $\overline{\gamma_{\Eve{}}}$ are the average received SNRs at \Bob{} and \Eve{}, respectively, $\alpha_{\ell,\texttt{X}}$ is the $\ell$th MFTR complex base-band fading amplitude, and \mbox{$\Omega_{\ell,\texttt{X}} = \mathbb{E} [\left|\alpha_{\ell,\texttt{X}}\right|^2]$} is the average channel power gain\footnote{The inclusion of $\Omega_{\ell,\texttt{X}}$ in the denominator of \eqref{eq:Total_SNR_Bob} is solely intended to normalize the fading mean power.}. 

From \cite[eq. (5)]{Vega_2023}, the instantaneous power of the $\ell$th MFTR fading amplitude in \eqref{eq:Total_SNR_Bob} can be written as\footnote{We have simplified the notation by omitting the subscript $\ell$.} 
\begin{align} \label{eq:R2}
    \nonumber \left| \alpha_{\texttt{X}} \right|^2 =& \left|\sqrt{g_{\texttt{X}}}\left(V_{1,\texttt{X}}\exp(\mathbbm{i}\phi_{1,\texttt{X}})
    + V_{2,\texttt{X}}\exp(\mathbbm{i}\phi_{2,\texttt{X}})\right) + \mathcal{G}_{1,\texttt{X}} \right|^2 \\
    &+ \sum_{n=2}^{\mu_{\texttt{X}}} |\sqrt{g_{\texttt{X}}}\, U_{n,\texttt{X}}\exp(\mathbbm{i}\psi_{n,\texttt{X}}) + \mathcal{G}_{n,\texttt{X}}|^2,
\end{align}
where $g_{\texttt{X}}$ is a unitary mean Gamma RV with PDF given by
\begin{align} \label{eq:PDF_zeta}
    f_{g_{\texttt{X}}}(z) = \frac{m_{\texttt{X}}^{m_{\texttt{X}}} z^{m_{\texttt{X}}-1}}{\Gamma{(m_{\texttt{X}})}}\exp(-m_{\texttt{X}} z),
\end{align}
where $m_{\texttt{X}}>0$ denotes the shadowing severity index of the specular wave components. In addition, the dominant specular components of the first arriving cluster have constant amplitudes $V_{1,\texttt{X}}$ and $V_{2,\texttt{X}}$, and uniformly distributed phases $\phi_{1,\texttt{X}}$ and $\phi_{2,\texttt{X}}$, respectively, such that $\phi_{y,\texttt{X}} \sim \mathcal{U}[0,2\pi)$, for $y\in\{1,2\}$. Moreover, $U_{n,\texttt{X}}\exp(\mathbbm{i}\psi_{n,\texttt{X}})$ represents the specular component of the $n$th cluster, where $U_{n,\texttt{X}}$ is a constant amplitude and $\psi_{n,\texttt{X}} \sim \mathcal{U}[0,2\pi)$. Finally, $\mathcal{G}_{n,\texttt{X}}$ is a circularly-symmetric RV such that $\mathcal{G}_{n,\texttt{X}}\sim \mathcal{CN}(0,2\sigma_{\texttt{X}}^2)$ for $n=1,2,...,\mu$, where $\mu$ is the number of multipath clusters.

The PDF of $\gamma_{\ell,\texttt{X}}$, defined in \eqref{eq:Total_SNR_Bob}, is given by \cite[eq. (20)]{Vega_2023}
\begin{align}
    \label{eq: marginal PDF}
    \nonumber f_{\gamma_{\ell,\texttt{X}}} (z) = &\sum_{i=0}^{\infty} \omega_{i,\texttt{X}}\frac{ \, z^{\mu_{\texttt{X}} + i -1} }{ \Gamma (\mu_{\texttt{X}} +i)}\left( \frac{\overline{\gamma_{\texttt{X}}}}{\mu_{\texttt{X}} \left(K_{\texttt{X}}+1 \right)} \right)^{-(\mu_{\texttt{X}} + i)} \\
    &\hspace{0.6cm}\times \exp \left(- \frac{  \mu_{\texttt{X}} \ z \left( K_{\texttt{X}}+1 \right)}{\overline{\gamma_{\texttt{X}}}} \right),
\end{align}
for $\texttt{X}\in\{\Bob{},\Eve{}\}$, where 
\begin{align}
    \label{eq: omega_i} 
    \nonumber \omega_{i,{\texttt{X}}} = & \frac{m_{{\texttt{X}}}^{m_{{\texttt{X}}}} \,\Gamma(m_{\texttt{X}}+i)  \left(\mu_{{\texttt{X}}} K_{{\texttt{X}}} (1- \Delta_{{\texttt{X}}}) \right)^{i}}{\sqrt{\pi} \, \Gamma (m_{{\texttt{X}}} ) \Gamma (i+1) \left(m_{{\texttt{X}}}+ \mu_{{\texttt{X}}} K_{{\texttt{X}}} (1- \Delta_{_{\texttt{X}}}) \right)^{m_{_{\texttt{X}}}+i}} \\
    & \times \sum^{i}_{q=0} \binom{i}{q} \frac{\Gamma(q+\frac{1}{2})}{\Gamma(q+1)} \left( \frac{2 \Delta_{{\texttt{X}}}}{1- \Delta_{{\texttt{X}}}} \right)^q \\
    \nonumber & \times \hspace{-0.1cm}\, _2F_1 \left( m_{{\texttt{X}}}+i,q+ \frac{1}{2}; q+1; -\frac{2 \mu_{{\texttt{X}}} K_{{\texttt{X}}} \Delta_{{\texttt{X}}}}{m_{{\texttt{X}}} + \mu_{{\texttt{X}}} K_{{\texttt{X}}} (1 - \Delta_{{\texttt{X}}})}\right).
\end{align}
In addition, in \eqref{eq: marginal PDF}, we have that
\begin{align} \label{eq:K}
    K_{\texttt{X}} = \frac{V_{1,\texttt{X}}^2 + V_{2,\texttt{X}}^2 + \sum_{n=2}^{\mu_{\texttt{X}}}U_{n,\texttt{X}}^{2}}{2\sigma^2 \mu_{\texttt{X}}},
\end{align}
is the ratio of the average power of the specular components to the power of the remaining scattered components, and
\begin{align} \label{eq:Delta}
    \Delta_{\texttt{X}} = \frac{2V_{1,\texttt{X}}V_{2,\texttt{X}}}{V_{1,\texttt{X}}^2 + V_{2,\texttt{X}}^2 + \sum_{n=2}^{\mu_{\texttt{X}}}U_{n,\texttt{X}}^{2}}, \quad 0\leq \Delta_{\texttt{X}} \leq 1,
\end{align}
represents the degree of similarity among the average received power of the dominant components within cluster 1.

The PDF of the instantaneous SNR at the output of an MRC receiver over MFTR fading is given by \cite{Garzon_2024}
\begin{align}
    \label{eq: sum PDF}
    f_{\Psi_{\textnormal{\textnormal{\texttt{X}}}}} (z) = \sum_{i=0}^{\infty} \varphi_{i,\textnormal{\texttt{X}}} \frac{z^{\nu_{i,\textnormal{\texttt{X}}}-1} \rho_{\textnormal{\texttt{X}}}^{-\nu_{i,\textnormal{\texttt{X}}}}}{\Gamma (\nu_{i,\textnormal{\texttt{X}}})} \exp \left(- \frac{z}{\rho_{\textnormal{\texttt{X}}}} \right),
\end{align}
where
\begin{equation}
    \label{eq: Coefficients}
    \varphi_{i,\textnormal{\texttt{X}}} = 
    \begin{cases}
    \omega_{0,\texttt{X}}^{L_{\textnormal{\texttt{X}}}} & i = 0\\
    \frac{1}{i \omega_{0,\texttt{X}}} \sum _{\ell=1}^i \varphi _{i-\ell,\textnormal{\texttt{X}}} \, \omega_{\ell,\textnormal{\texttt{X}}}  (L_{\textnormal{\texttt{X}}} \ell+\ell-i) & i \geq 1,
    \end{cases}
\end{equation}
\begin{align}
    \label{eq:rho}
        \rho_{\textnormal{\texttt{X}}}=  \frac{\overline{\gamma_{\textnormal{\texttt{X}}}}}{\mu_{\textnormal{\texttt{X}}} \left(K_{\textnormal{\texttt{X}}}+1 \right)},
\end{align}
and
\begin{align}
    \label{eq:nu}
        \nu_{i,\textnormal{\texttt{X}}}=i+ \mu_{\textnormal{\texttt{X}}} L_{\textnormal{\texttt{X}}}.
\end{align} 

\begin{proposition} \label{prop:CDF_SNR}
The CDF of the instantaneous SNR at the output of an MRC receiver over MFTR fading is given by
\begin{align}
    \label{eq: sum CDF}
    F_{\Psi_{\textbf{\texttt{X}}}} (z) = \sum_{i=0}^{\infty} \frac{\varphi_{i,\textnormal{\texttt{X}}} }{\Gamma (\nu_{i,\textnormal{\texttt{X}}})}\Upsilon \left(\nu_{i,\textnormal{\texttt{X}}},\frac{z}{\rho_{\textnormal{\texttt{X}}}} \right).
\end{align}
\end{proposition}
\begin{proof}
The CDF of $\Psi_{\texttt{X}}$ can be derived by integrating \eqref{eq: sum PDF} from zero to $z$ (i.e., $\int_{0}^{z} f_{\Psi_{\textnormal{\texttt{X}}}} (y) \text{d}y$), which yields \eqref{eq: sum CDF}.
\end{proof}

\section{Exact Secrecy Performance Analysis}
\label{sec:Exact_Analysis}

This section derives exact expressions for the GSOP, AFE, and AILR. 

\begin{definition} \label{def:FE}
    Fractional equivocation is a secrecy performance metric that quantifies the level of uncertainty that an eavesdropper has about the transmitted message. It is defined as the ratio between the eavesdropper’s equivocation rate and the transmission rate. For a given fading realization, the fractional equivocation can be written as \cite{Biao_2016}
    \begin{equation}
        \label{eq:Frac_Eq}
        \Lambda = \begin{cases} 
                    1, & C_{\Bob{}} - C_{\Eve{}} \geq R_{s} \\
                    \frac{1}{R_{s}}(C_{\Bob{}} - C_{\Eve{}}), & 0< C_{\Bob{}} - C_{\Eve{}}< R_{s}\\
                    0, &C_{\Bob{}} - C_{\Eve{}} \leq 0,
                  \end{cases}
    \end{equation}
    where $C_{\Bob{}}$ and $C_{\Eve{}}$ are the instantaneous channel capacities for the \Alice{}$\rightarrow$\Bob{} and \Alice{}$\rightarrow$\Eve{} links, respectively, and $R_{s}$ is defined as the secrecy rate. 
\end{definition}

    A value of FE close (or equal) to 1 indicates that the eavesdropper gains little (or no) information about the transmitted message, so that high (or perfect) secrecy is attained.
    In contrast, a value close (or equal) to 0 implies that most (or all) of the message has been compromised. 

We have that $C_{\Bob{}} = \log_{2}(1 + \Psi_{\Bob{}})$ and  $C_{\Eve{}} = \log_{2}(1 + \Psi_{\Eve{}})$. Then, it can be readily shown that $C_{\Bob{}}~-~ C_{\Eve{}}=\log_{2}\Phi$,~where
\begin{align}
    \label{eq:Phi}
    \Phi = \frac{1 + \Psi_{\Bob{}}}{1 + \Psi_{\Eve{}}}.
\end{align}

By using \eqref{eq:Phi}, the FE defined in \eqref{eq:Frac_Eq} can be rewritten as
\begin{equation}
    \label{eq:Frac_Eq_2}
    \Lambda = \begin{cases} 
                0, &\Phi \leq 1 \\
                \frac{1}{R_{s}}\log_{2}\Phi, &1 < \Phi < 2^{R_{s}}\\
                1, &\Phi \geq 2^{R_{s}}.
              \end{cases}
\end{equation}

Next, we derive the CDF of $\Phi$, given by \eqref{eq:Phi}, which will streamline subsequent analyses.
\begin{proposition}
\label{prop:CDF_Phi}
The CDF of $\Phi$ can be written as
\begin{align}
    \label{eq:Phi_CDF}
    \nonumber F_{\Phi} (z) =& \sum_{i=0}^{\infty} \varphi_{i,\Bob{}} \sum_{j=0}^{\infty} \varphi_{j,\Eve{}} \left(1 - \frac{\rho_{\Bob{}}^{\nu_{j,\Eve{}}}}{\Gamma(\nu_{j,\Eve{}})}\exp\left(- \frac{z-1}{\rho_{\Bob{}}}\right) \right.\\
    \nonumber &\times \sum_{a=0}^{\nu_{i,\Bob{}}-1} \frac{(z\rho_{\Eve{}})^{a} \Gamma(a + \nu_{j,\Eve{}})}{a! (\rho_{\Bob{}} + z\rho_{\Eve{}})^{a + \nu_{j,\Eve{}}}}\\
    &\times \left.\, _1F_1 \left(-a,1-a-\nu_{j,\Eve{}},\frac{(z-1)(\rho_{\Bob{}} + z\rho_{\Eve{}})}{z\rho_{\Bob{}}\rho_{\Eve{}}}\right)\right),
\end{align}
where $\varphi_{i,\textnormal{\texttt{X}}}$, $\rho_{\textnormal{\texttt{X}}}$, and $\nu_{i,\textnormal{\texttt{X}}}$, for  $\textnormal{\texttt{X}}\in\{\Bob{},\Eve{}\}$ are calculated using \eqref{eq: Coefficients}, \eqref{eq:rho}, and \eqref{eq:nu}, respectively.
\end{proposition}
\begin{proof}
The proof can be found in Appendix \ref{app:Proof_CDF_Phi}.
\end{proof}

\subsection{Generalized Secrecy Outage Probability (GSOP)} \label{subsec:GSOP}

\begin{definition}
    GSOP is defined as the probability of FE falling below a given threshold $\theta$, such that $0<\theta\leq 1$~\cite{Biao_2016}, i.e., 
    \begin{align}
        \label{eq:GSOP}
        \text{GSOP} = \text{P}(\Lambda < \theta).
    \end{align}
    The generalized secrecy outage probability can also be defined as the probability that the information leakage ratio (ILR), i.e., the percentage of confidential information disclosed to \Eve{}, $1-\Lambda$, exceeds a specified value, $1- \theta$.
\end{definition}

    FE quantifies the uncertainty that \Eve{} retains about the transmitted message, which is directly related to their decoding error probability. According to Fano’s inequality \cite{cover2006elements}, a higher fractional equivocation corresponds a higher decoding error probability at \Eve{}.
    Consequently, GSOP proves versatile for systems with diverse secrecy requirements, which are assessed through variations in \Eve{}'s capability to decode confidential messages, achieved by selecting different values of $\theta$. 

\begin{proposition} \label{prop:GSOP}
The exact generalized secrecy outage probability for a wireless system over MFTR fading channels where the legitimate user and the eavesdropper employ MRC at the receivers is given by
\begin{align}
    \label{eq:GSOP_2}
    \nonumber &\textnormal{GSOP} = \\
    \nonumber & \sum_{i=0}^{\infty} \varphi_{i,\Bob{}} \sum_{j=0}^{\infty} \varphi_{j,\Eve{}} \left(1 - \frac{\delta_{\theta}^{\nu_{j,\Eve{}}}}{\Gamma(\nu_{j,\Eve{}})}\exp\left(\frac{1}{\rho_{\Bob{}}} - \frac{1}{\delta_{\theta}\rho_{\Eve{}}}\right) \right.\\
    \nonumber &\times \sum_{a=0}^{\nu_{i,\Bob{}}-1} \frac{\Gamma(a + \nu_{j,\Eve{}})}{a!}\left(\delta_{\theta}+1\right)^{-(a + \nu_{j,\Eve{}})}\\
    &\times \left.\, _1F_1 \left(-a,1-a-\nu_{j,\Eve{}},\left(\delta_{\theta}+1\right)\left(\frac{1}{\delta_{\theta}\rho_{\Eve{}}} - \frac{1}{\rho_{\Bob{}}}\right)\right)\right),
\end{align}
where $\varphi_{i,\textnormal{\texttt{X}}}$, $\rho_{\textnormal{\texttt{X}}}$, and $\nu_{i,\textnormal{\texttt{X}}}$, for  $\textnormal{\texttt{X}}\in\{\Bob{},\Eve{}\}$ are calculated using \eqref{eq: Coefficients}, \eqref{eq:rho}, and \eqref{eq:nu}, respectively, and 
\begin{align}\label{eq:delta_theta}
    \delta_{\theta} = \frac{\rho_{\Bob{}}}{2^{\theta R_{s}}\rho_{\Eve{}}}.
\end{align}
\end{proposition}

\begin{proof}
The proof can be found in Appendix \ref{app:Proof_GSOP}.
\end{proof}

    In the special instance of $\theta = 1$, the GSOP given by \eqref{eq:GSOP_2} yields the classical SOP metric.

    It is important to indicate that the GSOP expression in \eqref{eq:GSOP_2} can be readily computed using standard mathematical software. Also, infinite summations can be efficiently truncated with a moderate number of terms, yielding highly accurate results. This aspect will be substantiated in Section \ref{sec:Num_Results}.

In \cite{Sharma_2023}, the partial secrecy regime is analyzed over FTR fading, which is a special case of the more general MFTR model used in our study. The authors explored a scenario with multiple antennas at \Eve{}, where it is shown that GSOP expressions consist of nested infinite summations, with the number of summations being proportional to the number of antennas at \Eve{}. In contrast, our GSOP expression, given by \eqref{eq:GSOP_2}, requires only two nested infinite summations, irrespective of the number of receiving antennas. Moreover, our analysis addresses a more intricate scenario by incorporating multiple receiving antennas at both \Bob{} and \Eve{}, as well as accounting for the possibility that the $\Alice{}\rightarrow\Bob{}$ and $\Alice{}\rightarrow\Eve{}$ links are subject to independently but not identically distributed fading channels.

\subsection{Average Fractional Equivocation (AFE)}

\begin{definition}
    AFE is an error-probability-based secrecy metric that serves as an asymptotic lower bound on the eavesdropper’s overall decoding error probability. Thus, a higher AFE indicates a higher average uncertainty at the eavesdropper, reflecting stronger secrecy performance.
    Given that the decoding error probability at the eavesdropper for a specific fading realization is asymptotically lower bounded by $\Lambda$, the average fractional equivocation is defined as the expected value of $\Lambda$ \cite{Biao_2016}, i.e., 
    \begin{align}
        \label{eq:AFE}
        \overline{\Lambda} = \mathbb{E}[\Lambda].
    \end{align}
\end{definition}

\begin{proposition} \label{prop:Proof_Exact_AFE}
An exact expression for calculating the average fractional equivocation for a wireless system over MFTR fading channels where the legitimate user and the eavesdropper employ MRC at the receivers is given by
\begin{align}
    \label{eq:AFE_Exact}
    \overline{\Lambda} = 1 - \frac{1}{\log(2^{R_{s}})} \int_{1}^{2^{R_{s}}} \frac{1}{z}F_{\Phi}(z)\text{d}z,
\end{align}
where $F_{\Phi}(z)$ is given by \eqref{eq:Phi_CDF} in Proposition \ref{prop:CDF_Phi}.
\end{proposition}

\begin{proof}
The proof can be found in Appendix \ref{app:Proof_Exact_AFE}.
\end{proof}


\subsection{Average Information Leakage Rate (AILR)}

\begin{definition}
    AILR is a metric that quantifies the rate at which confidential information is leaked to the eavesdropper over time. While different transmission schemes may result in similar SOPs, they can still exhibit considerable differences in the amount of leaked information. A lower AILR value indicates more robust PLS performance, implying reduced information leakage to the eavesdropper. Thus, the AILR serves as a complementary secrecy metric that captures these differences by evaluating the actual information exposure to \Eve{}. For a fixed-rate transmission scheme, the AILR is given by~\cite{Biao_2016}:
    \begin{align}
        \label{eq:AILR}
        R_{L} = \left(1 - \overline{\Lambda}\right)R_{s}.
    \end{align}
\end{definition}

\begin{corollary} \label{cor:AILR}
An exact expression for calculating the average information leakage rate for a wireless system over MFTR fading channels where the legitimate user and the eavesdropper employ MRC at the receivers is given by
\begin{align}
    \label{eq:AILR_Exact}
    R_{L} = \frac{1}{\log(2)} \int_{1}^{2^{R_{s}}} \frac{1}{z}F_{\Phi}(z)\text{d}z,
\end{align}
where $F_{\Phi}(z)$ is given by \eqref{eq:Phi_CDF} in Proposition \ref{prop:CDF_Phi}.
\end{corollary}

\begin{proof}
    The exact expression for computing the AILR can be obtained by substituting \eqref{eq:AFE_Exact} given in Proposition~\ref{prop:Proof_Exact_AFE} into \eqref{eq:AILR}, as formalized in Corollary~\ref{cor:AILR}, resulting in \eqref{eq:AILR_Exact}.
\end{proof}


Closed-form evaluation of \eqref{eq:AFE_Exact} and \eqref{eq:AILR_Exact} is feasible under specific assumptions, though these are impractical given the system parameters. Moreover, such an analysis requires incorporating more nested infinite summations.
Thus, we have omitted the evaluation of these expressions in closed form. Instead, we provide accurate approximations for computing the AFE and the AILR in Section \ref{sec:Approximate_Analysis}. 

\section{Approximate and Asymptotic Analysis}
\label{sec:Approximate_Analysis}

In this section, simpler analytical expressions are derived for the GSOP, AFE, and AILR metrics by focusing on the system performance at high SNR.

A high SNR approximation of $\Phi$, defined in \eqref{eq:Phi}, can be obtained considering that $\Psi_{\Bob{}} \gg 1$ and $\Psi_{\Eve{}} \gg 1$. In this case, we can define $\Phi \approx \Phi_{\textnormal{A}} = \Psi_{\Bob{}}/\Psi_{\Eve{}}$. The CDF of $\Phi_{\text{{A}}}$ is presented in the following proposition.

\begin{proposition} \label{prop:Proof_CDF_L}
The CDF of $\Phi_{\textnormal{A}}$ is given by 
\begin{align}
    \label{eq:CDF_Phi_Low}
    \nonumber &F_{\Phi_\textnormal{A}}(z) =\\
    &\sum_{i=0}^{\infty} \varphi_{i,\Bob{}} \sum_{j=0}^{\infty} \varphi_{j,\Eve{}} \left(1 - \frac{\rho_{\Bob{}}^{\nu_{j,\Eve{}}}}{\Gamma(\nu_{j,\Eve{}})} \sum_{a=0}^{\nu_{i,\Bob{}}-1} \frac{(z\rho_{\Eve{}})^{a} \Gamma(a + \nu_{j,\Eve{}})}{a! (\rho_{\Bob{}} + z\rho_{\Eve{}})^{a + \nu_{j,\Eve{}}}} \right),
\end{align}
where $\varphi_{i,\textnormal{\texttt{X}}}$, $\rho_{\textnormal{\texttt{X}}}$, and $\nu_{i,\textnormal{\texttt{X}}}$, for  $\textnormal{\texttt{X}}\in\{\Bob{},\Eve{}\}$ are calculated using \eqref{eq: Coefficients}, \eqref{eq:rho}, and \eqref{eq:nu}, respectively.
\end{proposition}
\begin{proof}
The proof is in Appendix \ref{app:Proof_CDF_L}.
\end{proof}

\subsection{Generalized Secrecy Outage Probability}

An approximation for the $\textnormal{GSOP}$ in the high SNR regime, denoted as $\textnormal{GSOP}_{\textnormal{A}}$, is presented in the next proposition. 

\begin{proposition} \label{prop:GSOP_Low}
An approximate expression for calculating the GSOP in a wireless system over MFTR fading channels where the legitimate user and the eavesdropper employ MRC at the receivers is given by
\begin{align}
    \label{eq:GSOP_Low}
    \nonumber \textnormal{GSOP}_{\textnormal{A}} =&
    \sum_{i=0}^{\infty} \varphi_{i,\Bob{}} \sum_{j=0}^{\infty} \varphi_{j,\Eve{}}
    \left(1 - \frac{\delta_{\theta}^{\nu_{j,\Eve{}}}}{\Gamma(\nu_{j,\Eve{}})} \right.\\
    &\left. \times \sum_{a=0}^{\nu_{i,\Bob{}}-1} \frac{\Gamma(a + \nu_{j,\Eve{}})}{a!} \left(\delta_{\theta}+1\right)^{-(a + \nu_{j,\Eve{}})}\right),
\end{align}
where $\delta_{\theta}$ is given by \eqref{eq:delta_theta} in Proposition \ref{prop:GSOP}.

\end{proposition}

\begin{proof}
The expression for $\textnormal{GSOP}_{\textnormal{A}}$ can be derived by employing a procedure analogous to that outlined in Appendix \ref{app:Proof_GSOP}, with the key distinction of considering $\Phi_{\textnormal{A}}$ instead of $\Phi$ in \eqref{eq:Frac_Eq_2}. Consequently, the analysis incorporates the CDF of $\Phi_{\textnormal{A}}$, $F_{\Phi_{\textnormal{A}}}(z)$, as defined in \eqref{eq:CDF_Phi_Low}, leading to \eqref{eq:GSOP_Low}.
\end{proof}

We now obtain an asymptotic GSOP expression assuming that $\overline{\gamma_{\Bob{}}} \rightarrow \infty$, as presented in Proposition \ref{prop:GSOP_Assympt}. 

\begin{proposition} \label{prop:GSOP_Assympt}
An asymptotic expression for the GSOP as $\overline{\gamma_{\Bob{}}}$ tends to infinity is given by
\begin{align}
    \label{eq:GSOP_Assymp}
    \textnormal{GSOP}^{\infty} = \mathcal{G}_{c}\left(\frac{1}{\,\overline{\gamma_{\Bob{}}}\,}\right)^{\mathcal{G}_{d}},
\end{align}
where $\mathcal{G}_{c}$ is the coding gain and it is given by
\begin{align}
    \label{eq:Gc}
    \nonumber \mathcal{G}_{c} =& \frac{\varphi_{0,\Bob{}}\left(\rho_{\Eve{}}\mu_{\Bob{}}(K_{\Bob{}}+1)2^{\theta R_{s}}\right)^{\mu_{\Bob{}}L_{\Bob{}}}}{\Gamma(\mu_{\Bob{}}L_{\Bob{}})\mu_{\Bob{}}L_{\Bob{}}}\sum_{j=0}^{\infty} \varphi_{j,\Eve{}} \left(\nu_{j,\Eve{}}\right)_{\mu_{\Bob{}}L_{\Bob{}}}\\
    &\times  \, _1F_1 \left(\mu_{\Bob{}}L_{\Bob{}},1-\mu_{\Bob{}}L_{\Bob{}}-\nu_{j,\Eve{}},\frac{1-2^{\theta R_{s}}}{\rho_{\Eve{}}}\right),
\end{align}
and $\mathcal{G}_{d}$ is the diversity gain given by
\begin{align}
    \label{eq:Gd}
    \mathcal{G}_{d} = \mu_{\Bob{}}L_{\Bob{}}.
\end{align}
\end{proposition}

\begin{proof}
The proof can be found in Appendix \ref{app:GSOP_Assympt}.
\end{proof}

The coding gain represents the margin of SNR improvement achieved through MRC while maintaining a constant diversity gain. On the other hand, the diversity gain represents the improvement in system performance resulting from the combination of multiple signal replicas at the receiver. Thus, from \eqref{eq:Gc} and \eqref{eq:Gd}, it is evident that as $\overline{\gamma_{\Bob{}}} \rightarrow \infty$, the coding gain in GSOP depends on the fading distribution parameters of the $\Alice{}\rightarrow \Bob{}$ and $\Alice{}\rightarrow \Eve{}$ links, the number of MRC branches at both $\Bob{}$ and $\Eve{}$, and the SNR at $\Eve{}$. Conversely, the diversity gain is solely determined by the number of MRC branches at $\Bob{}$, $L_{\Bob{}}$, and the total number of multipath clusters within the MFTR fading model at the $\Alice{}\rightarrow \Bob{}$ link, $\mu_{\Bob{}}$.

\subsection{Average Fractional Equivocation}

The following proposition introduces an approximation for the AFE in the high SNR regime.

\begin{proposition} \label{prop:AFE_Approx}
An approximate expression for calculating the average fractional equivocation in a wireless system over MFTR fading channels where the legitimate user and the eavesdropper employ MRC at the receivers is given by
\begin{align}
    \label{eq:AFE_Approx}
    \nonumber\overline{\Lambda}_{\textnormal{A}} =& 1 - \sum_{i=0}^{\infty}\varphi_{i,\Bob{}}\sum_{j=0}^{\infty} \varphi_{j,\Eve{}} \left(\vphantom{\int_{0_{0}}^{R_{s}^{2}}} 1 - \frac{(-1)^{\nu_{j,\Eve{}}}}{\log(2^{R_{s}})\Gamma(\nu_{j,\Eve{}})} \right. \\
    &\nonumber \times  \sum_{a=0}^{\nu_{i,\Bob{}}-1} \frac{\Gamma(a + \nu_{j,\Eve{}})}{a!} \left(\mathcal{B}\left(-\frac{\rho_{\Bob{}}}{\rho_{\Eve{}}},\nu_{j,\Eve{}},1-a-\nu_{j,\Eve{}} \right) \right.\\
    &\left. \left. -\mathcal{B}\left(-\frac{\rho_{\Bob{}}}{2^{R_{s}}\rho_{\Eve{}}},\nu_{j,\Eve{}},1-a-\nu_{j,\Eve{}} \right)   \right) \vphantom{\int_{0_{0}}^{R_{s}^{2}}} \right),
\end{align}
where $\varphi_{i,\textnormal{\texttt{X}}}$, $\rho_{\textnormal{\texttt{X}}}$, and $\nu_{i,\textnormal{\texttt{X}}}$, for  $\textnormal{\texttt{X}}\in\{\Bob{},\Eve{}\}$ are calculated using \eqref{eq: Coefficients}, \eqref{eq:rho}, and \eqref{eq:nu}, respectively.
\end{proposition}

\begin{proof}
The proof can be found in Appendix \ref{app:AFE_Approx}.
\end{proof}

\subsection{Average Information Leakage Rate}

An approximation for the AILR in the high SNR regime can be obtained from Proposition \ref{prop:AFE_Approx} and \eqref{eq:AILR}, as follows.
\begin{corollary} \label{cor:Approx_AILR}
An approximate expression for calculating the average information leakage rate in a wireless system over MFTR fading channels where the legitimate user and the eavesdropper employ MRC at the receivers is given by
\begin{align}
    \label{eq:AILR_Approx}
    \nonumber R_{L,\textnormal{A}} =&\, R_{s}\sum_{i=0}^{\infty}\varphi_{i,\Bob{}}\sum_{j=0}^{\infty} \varphi_{j,\Eve{}} \left(\vphantom{\int_{0_{0}}^{R_{s}^{2}}} 1 - \frac{(-1)^{\nu_{j,\Eve{}}}}{\log(2^{R_{s}})\Gamma(\nu_{j,\Eve{}})} \right. \\
    &\nonumber \times  \sum_{a=0}^{\nu_{i,\Bob{}}-1} \frac{\Gamma(a + \nu_{j,\Eve{}})}{a!} \left(\mathcal{B}\left(-\frac{\rho_{\Bob{}}}{\rho_{\Eve{}}},\nu_{j,\Eve{}},1-a-\nu_{j,\Eve{}} \right) \right.\\
    &\left. \left. -\mathcal{B}\left(-\frac{\rho_{\Bob{}}}{2^{R_{s}}\rho_{\Eve{}}},\nu_{j,\Eve{}},1-a-\nu_{j,\Eve{}} \right)   \right) \vphantom{\int_{0_{0}}^{R_{s}^{2}}} \right),
\end{align}
where $\varphi_{i,\textnormal{\texttt{X}}}$, $\rho_{\textnormal{\texttt{X}}}$, and $\nu_{i,\textnormal{\texttt{X}}}$, for  $\textnormal{\texttt{X}}\in\{\Bob{},\Eve{}\}$ are calculated using \eqref{eq: Coefficients}, \eqref{eq:rho}, and \eqref{eq:nu}, respectively.
\end{corollary}

\section{Computational Complexity and Truncation Error Analysis}
\label{sec:CompTrun_Analysis}

In this section, we present a computational complexity analysis of the proposed closed-form expressions, along with an evaluation of the truncation error introduced when the infinite series involved are truncated to a finite value $T$. The analysis is conducted for the exact and approximate GSOP expressions, which are selected as representative cases since the analytical approximate formulations of the AFE and AILR exhibit similar mathematical structures.

\subsection{Computational Complexity}
From \eqref{eq:GSOP_2}, the GSOP can be evaluated by truncating its infinite summations to a finite upper limit $T$, as follows
\begin{align}
    \label{eq:GSOP_T}
    \nonumber &\textnormal{GSOP} \approx \\
    \nonumber & \sum_{i=0}^{T} \varphi_{i,\Bob{}} \sum_{j=0}^{T} \varphi_{j,\Eve{}} \left(1 - \frac{\delta_{\theta}^{\nu_{j,\Eve{}}}}{\Gamma(\nu_{j,\Eve{}})}\exp\left(\frac{1}{\rho_{\Bob{}}} - \frac{1}{\delta_{\theta}\rho_{\Eve{}}}\right) \right.\\
    \nonumber &\times \sum_{a=0}^{\nu_{i,\Bob{}}-1} \frac{\Gamma(a + \nu_{j,\Eve{}})}{a!}\left(\delta_{\theta}+1\right)^{-(a + \nu_{j,\Eve{}})}\\
    &\times \left.\, _1F_1 \left(-a,1-a-\nu_{j,\Eve{}},\left(\delta_{\theta}+1\right)\left(\frac{1}{\delta_{\theta}\rho_{\Eve{}}} - \frac{1}{\rho_{\Bob{}}}\right)\right)\right).
\end{align}

The overall computational complexity of evaluating \eqref{eq:GSOP_T} grows polynomially with the dominant contribution arising from its three nested summations. The computation begins with the pre-calculation of the recursive coefficients $\varphi_{i,\Bob{}}$ and $\varphi_{j,\Eve{}}$ for $i = 0, \dots, T$. For $i = 0$, $\varphi_{0,_{\texttt{X}}} = \omega_{0,\texttt{X}}^{L_{\texttt{X}}}$, which requires $\mathcal{O}(\log L_{\texttt{X}})$ multiplications via binary exponentiation. For $i \geq 1$, each $\varphi_{i,x}$ involves a summation with $i$ terms, where each term requires one addition, one subtraction, and three multiplications, plus two multiplications and one division outside the summation. The total cost for computing each $\varphi_{i,x}$ is therefore $i$ additions, $i$ subtractions, and $3i+2$ multiplications. Summing over $i = 1, \dots, T$ yields $\mathcal{O}(T^2)$ additions, subtractions, and multiplications for the complete pre-computation.

However, the main cost comes from the triple summation, where the two outer loops (indices $i$ and $j$) each iterate up to the truncation value $T$. Within these, the innermost loop over index $a$ has complexity $\mathcal{O}(i^3)$. This is because its limit, $\nu_{i,B} - 1$, is proportional to $i$, meaning the loop runs approximately $i$ times. The most computationally demanding operation inside this loop is the evaluation of the Kummer confluent hypergeometric function, $_1F_1(-a, \cdot, \cdot)$. Here, the first argument $-a$ is a negative integer, so the series expansion terminates after $a+1$ terms. Computing this finite series requires $\mathcal{O}(a^2)$ operations, and summing over $a = 0$ to $\nu_{i,\Bob{}} - 1$ yields $\mathcal{O}(i^3)$. When aggregated over all $i = 0$ to $T$, the complexity becomes $\mathcal{O}(T^4)$. Since the $j$-loop repeats this process $T$ times, we have that the total computational complexity evaluation of \eqref{eq:GSOP_T} amounts to $\mathcal{O}(T^5)$.

The computation of the GSOP can be simplified further by considering the high-SNR regime ($\overline{\gamma_{\Bob{}}} \rightarrow \infty$). In this case, we can employ the asymptotic expression given in \eqref{eq:GSOP_Assymp}. If the summation is truncated to a finite value $T$, the resulting expression becomes
\begin{align}
    \label{eq:GSOP_Assymp_trunc}
    \nonumber&\textnormal{GSOP}^{\infty} \approx\\ \nonumber&\frac{\varphi_{0,\Bob{}}\left(\rho_{\Eve{}}\mu_{\Bob{}}(K_{\Bob{}}+1)2^{\theta R_{s}}\right)^{\mu_{\Bob{}}L_{\Bob{}}}}{\Gamma(\mu_{\Bob{}}L_{\Bob{}})\mu_{\Bob{}}L_{\Bob{}}}\sum_{j=0}^{T} \varphi_{j,\Eve{}} \left(\nu_{j,\Eve{}}\right)_{\mu_{\Bob{}}L_{\Bob{}}}\\
    &\times  \, _1F_1 \left(\mu_{\Bob{}}L_{\Bob{}},1-\mu_{\Bob{}}L_{\Bob{}}-\nu_{j,\Eve{}},\frac{1-2^{\theta R_{s}}}{\rho_{\Eve{}}}\right)\left(\frac{1}{\,\overline{\gamma_{\Bob{}}}\,}\right)^{\mu_{\Bob{}}L_{\Bob{}}}.
\end{align}
The total computational complexity for calculating \eqref{eq:GSOP_Assymp_trunc} is primarily determined by the single infinite summation over index $j$. This summation is truncated at a value $T$ and has two main sources of complexity. First, the recursive coefficients $\varphi_{j,\Eve{}}$ must be pre-computed for all values up to $T$ using the recursive formula, which has a computational cost of $\mathcal{O}(T^2)$. Second, within the summation loop that runs from $j=0$ to $T$, the most computationally intensive task is evaluating the Kummer's confluent hypergeometric function, $_1F_1(\cdot,\cdot,\cdot)$. Since its first parameter, $\mu_{\Bob{}} L_{\Bob{}}$, is a positive value and thus not a non-positive integer, its series does not terminate naturally and must be truncated, for instance, at a value $T$, leading to a complexity of $\mathcal{O}(T^2)$ per evaluation. Therefore, the total complexity of the summation is the number of iterations times the complexity per iteration, resulting in $\mathcal{O}(T \cdot T^2) = \mathcal{O}(T^3)$. Thus, the overall complexity of evaluating \eqref{eq:GSOP_Assymp_trunc} is the sum of these two components, which simplifies to $\mathcal{O}(T^3)$.

It is worth noting that, unlike other approaches in the literature where secrecy performance metrics involve deeply nested summations, whose complexity can grow exponentially with the number of antennas at \Bob{} or \Eve{}, the proposed formulation maintains a polynomial complexity that is independent of the antenna array size. This makes our expressions more scalable and computationally efficient for multi-antenna scenarios.

\subsection{Truncation Error}

In general, obtaining a closed-form expression for the truncation error of the series in~\eqref{eq:GSOP_T} or \eqref{eq:GSOP_Assymp_trunc} is not feasible due to the intricate structure of the summation terms and their dependence on multiple system parameters. To illustrate this and, for the sake of simplicity, we focus on the approximate GSOP expression in \eqref{eq:GSOP_Assymp_trunc}, where the infinite series is truncated to a finite value $T$. In this case, let $a_j$ denote the $j$-th term of the truncated series, i.e.,
\begin{align}\label{eq:a_j}
\nonumber a_j = &\varphi_{j,\Eve{}} \left(\nu_{j,\Eve{}}\right)_{\mu_{\Bob{}}L_{\Bob{}}}\\
    &\times  \, _1F_1 \left(\mu_{\Bob{}}L_{\Bob{}},1-\mu_{\Bob{}}L_{\Bob{}}-\nu_{j,\Eve{}},\frac{1-2^{\theta R_{s}}}{\rho_{\Eve{}}}\right),
\end{align}
where $\varphi_{j,\Eve{}}$ is defined recursively from the coefficients $\omega_{i,\Eve{}}$ given in~\eqref{eq: omega_i}. These coefficients involve gamma functions, finite combinatorial sums, and Gauss hypergeometric functions ${}_2F_1(\cdot,\cdot;\cdot;\cdot)$ whose parameters depend explicitly on the summation index $j$. As a result, the general term $a_j$ cannot be re-expressed as a standard hypergeometric function in $j$ for which the remainder from $T+1$ to infinity admits a known closed-form representation. Furthermore, the presence of parameters $(m_{\Eve{}}, \mu_{\Eve{}}, K_{\Eve{}}, \Delta_{\Eve{}})$ that may take arbitrary real positive values prevents simplification to special parameter cases where such a closed form could exist.

In principle, a geometric bound for the truncation error could be derived from the structure of $\omega_{i,\Eve{}}$, as it exhibits a dominant geometric factor when the sufficient condition $m_{\Eve{}} > \mu_{\Eve{}} K_{\Eve{}} (1 + \Delta_{\Eve{}})$ is satisfied. However, since this condition is not guaranteed for all parameter configurations of interest, such a bound is not a viable general solution. Instead, Appendix \ref{app:Ratio_Bound} shows that the truncation error can be estimated using 
\begin{align}\label{eq:ratio_bound}
|\mathcal{R}_T| & \le \frac{\varphi_{0,\Bob{}}\left(\rho_{\Eve{}}\mu_{\Bob{}}(K_{\Bob{}}+1)2^{\theta R_{s}}\right)^{\mu_{\Bob{}}L_{\Bob{}}}|a_{T+1}|}{\Gamma(\mu_{\Bob{}}L_{\Bob{}})\mu_{\Bob{}}L_{\Bob{}}(1-r_{T+1})} \left(\frac{1}{\,\overline{\gamma_{\Bob{}}}\,}\right)^{\mu_{\Bob{}}L_{\Bob{}}},
\end{align}
where $a_{j}$ is given in \eqref{eq:a_j} and 
\begin{align}\label{eq:rj}
r_{j} &= \frac{|a_{j+1}|}{|a_{j}|},
\end{align}
which depends only on the last two computed terms of the series and adapts to any set of parameters. 

\section{Numerical Results and Discussions}
\label{sec:Num_Results}

In this section, we present numerical results and discussions regarding the secrecy performance of the analyzed system across diverse representative scenarios. Monte Carlo simulations validate the accuracy of the analytical expressions. 

Theoretical curves in the following figures are computed by truncating the infinite sums in the analytical expressions to 50 terms. Additionally, the Monte Carlo simulations were carried out in \textsc{Matlab} using $8 \times 10^7$ iterations, where in each iteration $10^4$ random samples were generated, thereby corresponding to a total of $8 \times 10^{11}$ realizations. The system parameters considered in each figure were chosen to illustrate representative scenarios rather than to reproduce any specific measurement campaign. These parameters are consistent with the ranges reported in~\cite[Table I]{Vega_2023}, which are known to emulate a variety of fading conditions. In particular, different parameter combinations can reproduce well-known fading distributions such as Rayleigh, Nakagami-$m$, Rician, $\kappa$-$\mu$, and FTR, among others. 
Finally, $R_s$ is set to $1$ bit/s/Hz unless stated otherwise.

\subsection{GSOP}

In the following, the exact theoretical GSOP is determined using \eqref{eq:GSOP_2}, the approximate GSOP is plotted employing \eqref{eq:GSOP_Low}, and the asymptotic GSOP curves are generated using \eqref{eq:GSOP_Assymp}.

Fig. \ref{fig:GSOP_fEsNob_parLb_Le_EsNoe_8dB_Theta05} shows the GSOP versus $\overline{\gamma_{\Bob{}}}$, parameterized by $L_{\Bob{}}$ and $L_{\Eve{}}$, considering $m_{\Bob{}}=4, \mu_{\Bob{}} = 2, \sigma_{\Bob{}}^{2}=0.5, K_{\Bob{}}= 2.5, \Delta_{\Bob{}}=0.8$, $\overline{\gamma_{\Eve{}}} = 8$ dB, $m_{\Eve{}}=3, \mu_{\Eve{}} = 2, \sigma_{\Eve{}}^{2}=0.5, K_{\Eve{}}= 6.5, \Delta_{\Eve{}}=0.9$ and $\theta = 0.5$. In the figure, observe the accuracy of both the exact and approximate theoretical curves in comparison to the simulations for any $\overline{\gamma_{\Bob{}}}$, and notice the accuracy of the asymptotic curves in the high SNR region. These findings validate the assertions made in Subsection \ref{subsec:GSOP}. Note that a diversity gain is achieved as $L_{\Bob{}}$ increases. Thus, the slope of the GSOP curves changes, leading to an improvement in security terms, i.e., a decrease in GSOP for a given $\overline{\gamma_{\Bob{}}}$. On the other hand, notice that as the number of antennas at \Eve{} increases, the GSOP increases. However, this decline in security performance is manifested as a loss in SNR, yet the slope of the GSOP curves remains unchanged. This implies a loss in terms of coding gain. Hence, increasing $L_{\Bob{}}$ ensures a significantly greater enhancement in secrecy performance compared to the secrecy performance loss that occurs when $L_{\Eve{}}$ increases. These results corroborate the findings stated in Section \ref{sec:Approximate_Analysis}.

\begin{figure}[t]
    \centerline{\includegraphics[width=\linewidth]
    {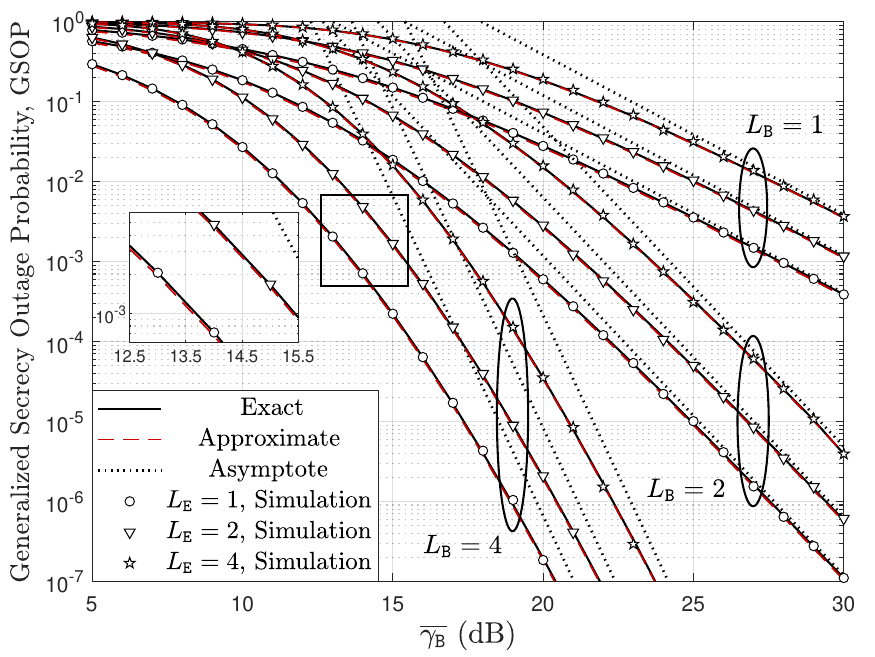}}
    \caption{GSOP as a function of $\overline{\gamma_{\Bob{}}}$, parameterized by $L_{\Bob{}}$ and $L_{\Eve{}}$ considering $m_{\Bob{}}=4, \mu_{\Bob{}} = 2, \sigma_{\Bob{}}^{2}=0.5, K_{\Bob{}}= 2.5, \Delta_{\Bob{}}=0.8$, $\overline{\gamma_{\Eve{}}} = 8$ dB, $m_{\Eve{}}=3, \mu_{\Eve{}} = 2, \sigma_{\Eve{}}^{2}=0.5, K_{\Eve{}}= 6.5$, and $\Delta_{\Eve{}}=0.9$, and $\theta = 0.5$.}
    \label{fig:GSOP_fEsNob_parLb_Le_EsNoe_8dB_Theta05}
    \vspace{-3mm}
\end{figure}

\begin{figure}[t]
    \centerline{\includegraphics[width=\linewidth]
    {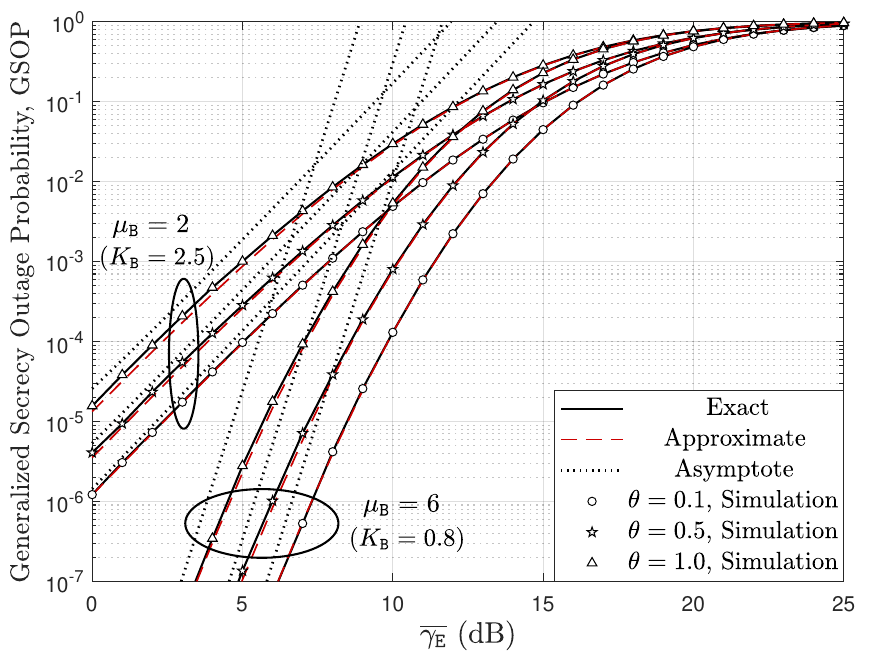}}
    \caption{GSOP as a function of $\overline{\gamma_{\Eve{}}}$, parameterized by $\mu_{\Bob{}}$, $K_{\Bob{}}$, and $\theta$ considering $L_{\Bob{}}=L_{\Eve{}}= 2, \overline{\gamma_{\Bob{}}} = 20$ dB, $m_{\Bob{}}=4, \sigma_{\Bob{}}^{2}=0.5, \Delta_{\Bob{}}=0.8, m_{\Eve{}}=3, \mu_{\Eve{}} = 2, \sigma_{\Eve{}}^{2}=0.5, K_{\Eve{}}= 6.5$, and $\Delta_{\Eve{}}=0.9$.}
    \label{fig:GSOP_fEsNoe_parTheta_mu_LbLe2_EsNob_20dB}
    \vspace{-3mm}
\end{figure}

Fig. \ref{fig:GSOP_fEsNoe_parTheta_mu_LbLe2_EsNob_20dB} shows the GSOP as a function of $\overline{\gamma_{\Eve{}}}$, parameterized by $\mu_{\Bob{}}$, $K_{\Bob{}}$, and $\theta$, assuming $L_{\Bob{}}=L_{\Eve{}}= 2, \overline{\gamma_{\Bob{}}} = 20$ dB, $m_{\Bob{}}=4, \sigma_{\Bob{}}^{2}=0.5, \Delta_{\Bob{}}=0.8, m_{\Eve{}}=3, \mu_{\Eve{}} = 2, \sigma_{\Eve{}}^{2}=0.5,  K_{\Eve{}}= 6.5$, and $\Delta_{\Eve{}}=0.9$. As evident from \eqref{eq:K}, altering $\mu_{\Bob{}}$ while keeping the other parameters constant ($V_{1,\Bob{}},V_{2,\Bob{}}$, and $U_{n,\Bob{}}$) leads to a change in $K_{\Bob{}}$. For this reason, the figure displays varying values of $\mu_{\Bob{}}$ and $K_{\Bob{}}$. In the results, notice that as $\overline{\gamma_{\Eve{}}}$ increases, the GSOP also increases. As $\overline{\gamma_{\Eve{}}}$ rises, it can be interpreted as $\Eve{}$ being closer to $\Alice{}$, thereby having improved reception power and, consequently, increased channel capacity. This represents a deterioration in terms of secrecy performance. On the other hand, when $\mu_{\Bob{}}$ increases, an improvement in secrecy diversity gain is achieved, according to \eqref{eq:Gd}, implying a gain in terms of secrecy performance, or equivalently, a lower GSOP. Hence, a higher number of multipath clusters in the $\Alice{}\rightarrow\Bob{}$ link fading channel implies better propagation conditions for \Bob{}, resulting in improved secrecy performance. 
Notice that, to achieve equivalent performance in terms of GSOP, a difference of approximately 3 dB in $\overline{\gamma_{\Eve{}}}$ is necessary when comparing the scenario where \Eve{}'s ability to decode the confidential message is very low ($\theta=0.1$) with that of the classical SOP approach ($\theta=1$). In scenarios where it is possible to relax the secrecy requirement due to a lower capability of \Eve{}, greater proximity of \Eve{} to $\Alice{}$ can be tolerated, or equivalently, a higher value of $\overline{\gamma_{\Eve{}}}$ can be endured if GSOP is considered as a design criterion.

Fig. \ref{fig:GSOP_fRs_parKb_mb_Lb3Le2_EsNob25_EsNoe15} shows the GSOP as a function of $R_{s}$, parameterized by $m_{\Bob{}}$ and $K_{\Bob{}}$, considering $L_{\Bob{}}=3$, $L_{\Eve{}}= 2, \overline{\gamma_{\Bob{}}} = 25$ dB, $\overline{\gamma_{\Eve{}}} = 15$ dB, $\Delta_{\Bob{}}=0.1, m_{\Eve{}}=2, \mu_{\Eve{}} = 2, K_{\Eve{}}= 4$, $\Delta_{\Eve{}}=0.3$, and $\theta=0.5$. 
Note from Fig.~\ref{fig:GSOP_fRs_parKb_mb_Lb3Le2_EsNob25_EsNoe15} that as the target secrecy rate $R_s$ increases, the GSOP also increases, which is an expected result. This trend can be understood by analyzing the behavior of FE, $\Lambda$, given by \eqref{eq:Frac_Eq}, which quantifies the level of uncertainty that \Eve{} has about the transmitted message. Thus, FE depends on the difference between the instantaneous capacities of the main and wiretap links, $C_B - C_E$, relative to $R_s$.  As $R_s$ increases, the condition $C_B - C_E \geq R_s$ becomes increasingly difficult to satisfy, reducing the probability that the FE reaches its maximum value of 1. Consequently, the probability that FE will fall below a given threshold $\theta$ increases. Since the GSOP is defined as the probability that $\Lambda < \theta$ (refer to \eqref{eq:GSOP}), an increase in $R_s$ naturally leads to an increase in GSOP. 

In addition, in Fig \ref{fig:GSOP_fRs_parKb_mb_Lb3Le2_EsNob25_EsNoe15} notice that as $m_{\Bob{}}$ increases, the GSOP decreases for a given value of $K_{\Bob{}}$. Specifically, as $m_{\Bob{}}$ increases, the severity of the shadowing affecting the specular waves for the link $\Alice{}\rightarrow\Bob{}$ is reduced, which means better propagation conditions for \Bob{}, and consequently, a greater level of security. Interestingly, note from Fig. \ref{fig:GSOP_fRs_parKb_mb_Lb3Le2_EsNob25_EsNoe15} that, for $m_{\Bob{}}=1$, the GSOP decreases as $K_{\Bob{}}$ decreases. Conversely, for $m_{\Bob{}}=20$, the GSOP decreases as $K_{\Bob{}}$ increases. At this point, it is important to indicate that, when $m_{\Bob{}}$ is small, the shadowing that affects the specular components of the signal received at \Bob{} is highly severe. In other words, the variance of $g_{\Bob{}}$ in \eqref{eq:R2} around its mean value of $1$ is large. Additionally, an increase in $K_{\Bob{}}$ implies that the specular components $V_{1,\Bob{}}$ and $V_{2,\Bob{}}$ increase, provided that $\sigma_{\Bob{}}^2$ and $\mu_{\Bob{}}$ remain constant (refer to \eqref{eq:K}). Hence, since $V_{1,\Bob{}}$ and $V_{2,\Bob{}}$ are multiplied by $g_{\Bob{}}$, if these amplitudes increase, there is also an increase in the variance of the received signal around its mean value. This results in even more severe shadowing as $K_{\Bob{}}$ increases, which, in turn, leads to a higher GSOP. On the other hand, when $m_{\Bob{}}$ is large, the variance of $g_{\Bob{}}$ decreases. In this case, an increase in $V_{1,\Bob{}}$ and $V_{2,\Bob{}}$ directly enhances the power of the specular components, which remain unaffected by the randomness of shadowing. Therefore, the $\Alice{}\rightarrow \Bob{}$ link experiences improved propagation conditions, resulting in a lower GSOP.

\begin{figure}[t]
    \centerline{\includegraphics[width=\linewidth]
    {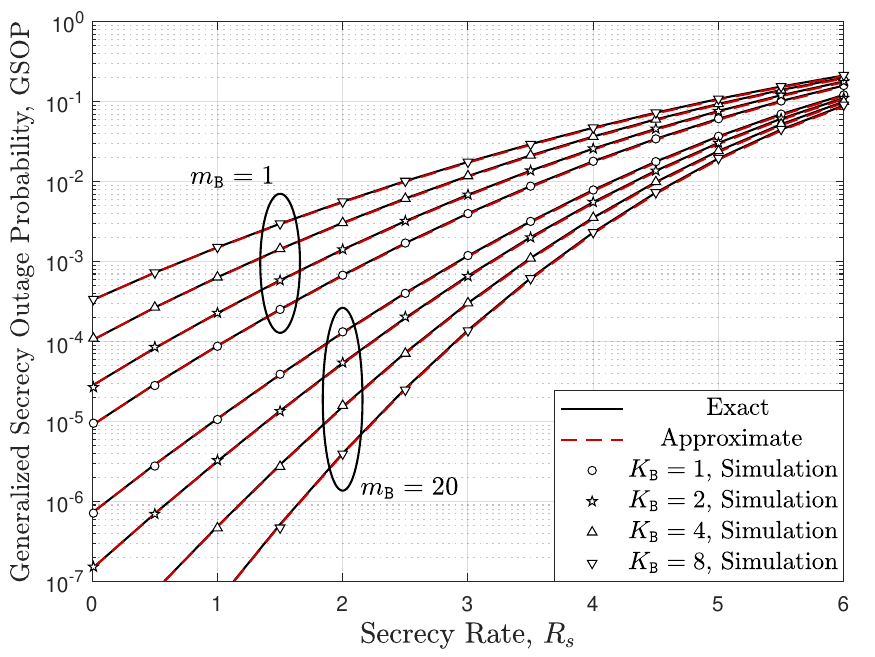}}
    \caption{GSOP as a function of $R_{s}$, parameterized by $m_{\Bob{}}$ and $K_{\Bob{}}$ considering $L_{\Bob{}}=3$, $L_{\Eve{}}= 2, \overline{\gamma_{\Bob{}}} = 25$ dB, $\Delta_{\Bob{}}=0.1, \sigma_{\Bob{}}^{2}=0.5$, $\overline{\gamma_{\Eve{}}} = 15$ dB, $m_{\Eve{}}=2, \mu_{\Eve{}} = 2, \sigma_{\Eve{}}^{2}=0.5, K_{\Eve{}}= 4$, $\Delta_{\Eve{}}=0.3$, and $\theta=0.5$.}
    \label{fig:GSOP_fRs_parKb_mb_Lb3Le2_EsNob25_EsNoe15}
    \vspace{-3mm}
\end{figure}

Fig. \ref{fig:GSOPSim_fEsNob_parLb_Le1_EsNoe10_Theta1_MFTR_Chan} shows the GSOP as a function of $\overline{\gamma_{\Bob{}}}$, parameterized by $L_{\Bob{}}$ and the fading channel experiencing by \Eve{}, considering $m_{\Bob{}}=3, \mu_{\Bob{}} = 2, \sigma_{\Bob{}}^{2}=0.5, K_{\Bob{}}= 2.5, \Delta_{\Bob{}}=0.8$, $L_{\Eve{}}=1$, the fading parameters in Table \ref{tab:Fad_Par_Eve}, and $\theta=1$. Thus, utilizing the versatility of the MFTR distribution as a comprehensive model for simulating various fading channel scenarios, the figure explores a scenario where the $\Alice{}\rightarrow\Bob{}$ link experiences MFTR fading, and the $\Alice{}\rightarrow\Eve{}$ link experiences commonly used fading models, including Rayleigh and Rician, as well as those pertinent to mmWave, sub-THz and THz propagation environments, such as $\kappa$-$\mu$ and FTR, respectively. This figure seeks to illustrate how changes in the fading channel parameters of the $\Alice{}\rightarrow\Eve{}$ link affect secrecy performance.
Firstly, it is noteworthy that, when $\Alice{}\rightarrow\Eve{}$ link is over Rayleigh fading, the highest GSOP is observed for both $L_{\Bob{}}=1$ and $L_{\Bob{}}=3$. Conversely, in the case of Rician fading, there is an intriguing observation: the GSOP decreases compared to the first scenario, and this reduction becomes more pronounced with an increase of the $K_{\Eve{}}$ factor. Moreover, the decrease in GSOP is more substantial for $L_{\Bob{}}=3$ than for $L_{\Bob{}}=1$ (the region enclosed in the square has been magnified for better visibility of the results). Note also that the second-to-last minimal GSOP occurs when the $\Alice{}\rightarrow\Eve{}$ link is subject to $\kappa$-$\mu$ fading. Furthermore, when this link experiences FTR fading, the GSOP approximates that observed under Rayleigh fading.

\begin{figure}[t]
    \centerline{\includegraphics[width=\linewidth]
    {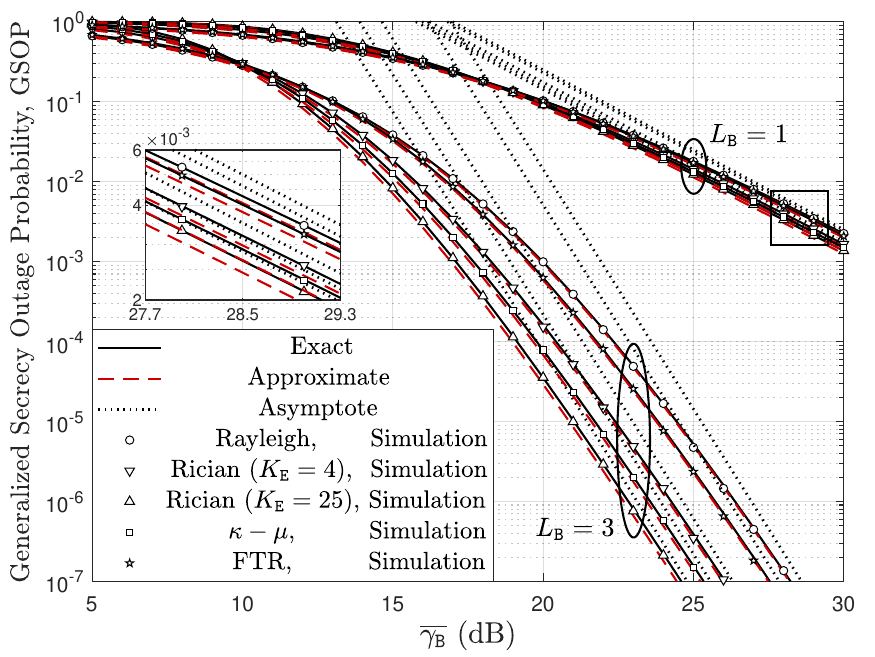}}
    \caption{GSOP as a function of $\overline{\gamma_{\Bob{}}}$, parameterized by $L_{\Bob{}}$ and the fading channel experiencing by \Eve{}, considering $m_{\Bob{}}=3, \mu_{\Bob{}} = 2, \sigma_{\Bob{}}^{2}=0.5, K_{\Bob{}}= 2.5, \Delta_{\Bob{}}=0.8$, $L_{\Eve{}}=1$, the fading parameters in Table \ref{tab:Fad_Par_Eve}, and $\theta=1$.}
    \label{fig:GSOPSim_fEsNob_parLb_Le1_EsNoe10_Theta1_MFTR_Chan}
    \vspace{-1mm}
\end{figure}

\begin{table}[t]
\centering
\scriptsize
\caption{Fading parameters for the $\Alice{}\rightarrow\Eve{}$ link used for obtaining the results shown in Fig. \ref{fig:GSOPSim_fEsNob_parLb_Le1_EsNoe10_Theta1_MFTR_Chan}.}
\begin{tabular}{@{}cccccc@{}}
\toprule
\textbf{Fading Distribution} & $m_{\Eve{}}$ & $\mu_{\Eve{}}$ & $\sigma_{\Eve{}}^2$ & $K_{\Eve{}}$ & $\Delta_{\Eve{}}$ \\
\midrule
Rayleigh & $100$ & $1$ & $0.5$ & $10^{-4}$ & $0$ \\
\midrule
Rician & $100$ & $1$ & $0.5$ & $4$ & $0$ \\
\midrule
Rician & $100$ & $1$ & $0.5$ & $25$ & $0$ \\
\midrule
$\kappa$-$\mu$ & $100$ & $3$ & $0.5$ & $1.6$ & $0$ \\
\midrule
FTR & $5$ & $1$ & $0.5$ & $3.3$ & $0.9$ \\
\bottomrule
\end{tabular}
\label{tab:Fad_Par_Eve}
\end{table}

To comprehend the results described in the previous paragraph, let us examine the classical definition of SOP, which can be derived from GSOP, given by \eqref{eq:GSOP_2}, employing $\theta = 1$ (i.e., perfect secrecy, as assumed in Fig. \ref{fig:GSOPSim_fEsNob_parLb_Le1_EsNoe10_Theta1_MFTR_Chan}). In this case, the SOP is obtained as 
\begin{align}
    \label{eq:SOP}
    \nonumber \textnormal{SOP} &= \text{P}\left(\Phi \leq 2^{R_{s}}\right) \\
    \nonumber&=\text{P}\left(\Psi_{\Bob{}}\right. \leq \underbrace{2^{R_{s}} + 2^{R_{s}}\Psi_{\Eve{}} -1}_{\Psi_{\Eve{},\text{eq}}} \left. \vphantom{2^{R_{s}}}\right)\\
    &=\int_{0}^{\infty} F_{\Psi_{\Bob{}}}\left(2^{R_{s}} + 2^{R_{s}}\Psi_{\Eve{}} -1\right) f_{\Psi_{\Eve{},\text{eq}}} (\Psi_{\Eve{}}) \text{d}\Psi_{\Eve{}},
\end{align}
where we have used the definition of $\Phi$, given by \eqref{eq:Phi}, $F_{\Psi_{\Bob{}}}(z)$ is the CDF of $\Psi_{\Bob{}}$, given by \eqref{eq: sum CDF}, and it can be shown that the PDF of $\Psi_{\Eve{},\text{eq}}$ is obtained as
\begin{align}
    \label{eq:PDF_Psieq}
    f_{\Psi_{\Eve{},\text{eq}}}(z) = 2^{-R_{s}}f_{\Psi_{\Eve{}}}\left((z+1)2^{-R_{s}} - 1\right),
\end{align}
where $f_{\Psi_{\Eve{}}}(z)$ is given by \eqref{eq: sum PDF}. 

\begin{figure*}[t]
  \centering
  \begin{tabular}[c]{cc}
    \begin{subfigure}[c]{0.3\textwidth}
      \includegraphics[width=4.7cm, height=5.2cm, left]{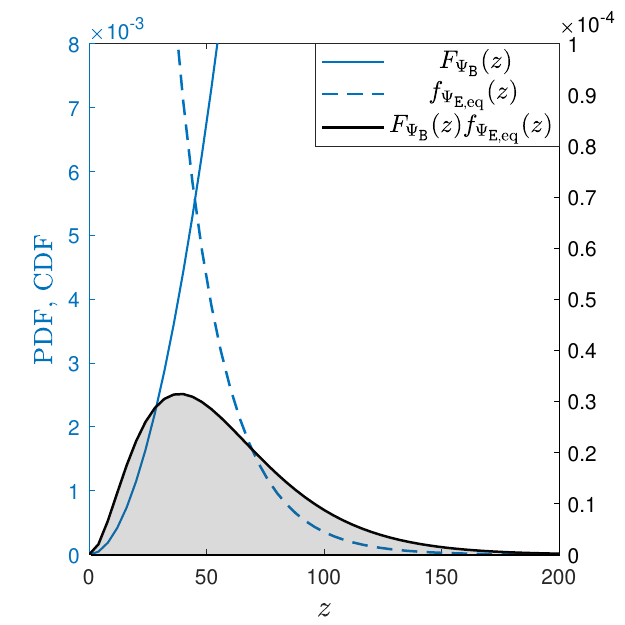}
      \caption{$L_{\Bob{}}=1$, Channel \Eve{}: Rayleigh.}
      \label{fig:MFTR_Ray_Lb1Le1_PDFCDFs}
    \end{subfigure}

    \begin{subfigure}[c]{0.3\textwidth}
      \includegraphics[width=4.7cm, height=5.2cm, left]{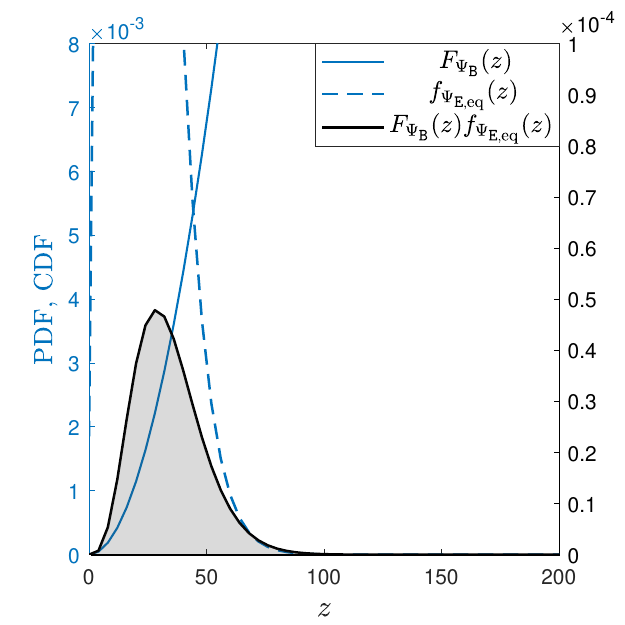}
      \caption{$L_{\Bob{}}=1$, Channel \Eve{}: Rician, $K_{\Eve{}}=4$.}
      \label{fig:MFTR_RicK4_Lb1Le1_PDFCDFs}
    \end{subfigure}
    \vspace{0.2cm}

    \begin{subfigure}[c]{0.3\textwidth}
      \includegraphics[width=4.7cm, height=5.2cm, left]{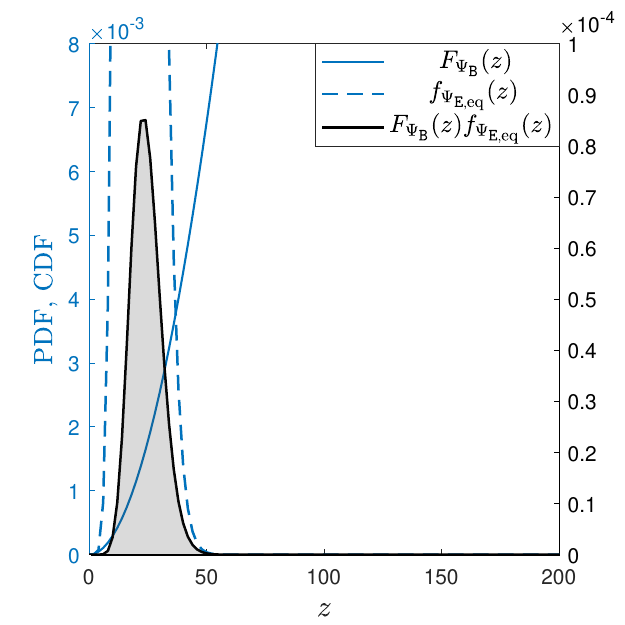}
      \caption{$L_{\Bob{}}=1$, Channel \Eve{}: Rician, $K_{\Eve{}}=25$.}
      \label{fig:MFTR_RicK25_Lb1Le1_PDFCDFs}
    \end{subfigure}
        \\
    \begin{subfigure}[c]{0.3\textwidth}
      \includegraphics[width=4.7cm, height=5.2cm, left]{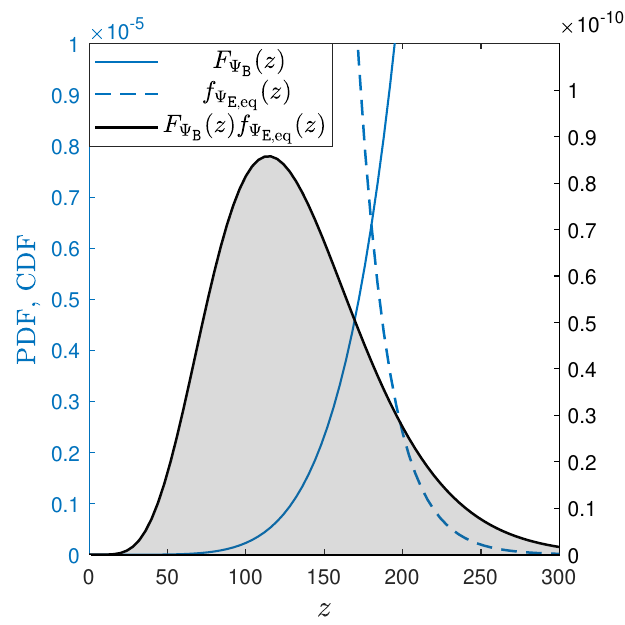}
      \caption{$L_{\Bob{}}=3$, Channel \Eve{}: Rayleigh.}
      \label{fig:MFTR_Ray_Lb3Le1_PDFCDFs}
    \end{subfigure}

    \begin{subfigure}[c]{0.3\textwidth}
      \includegraphics[width=4.7cm, height=5.2cm, left]{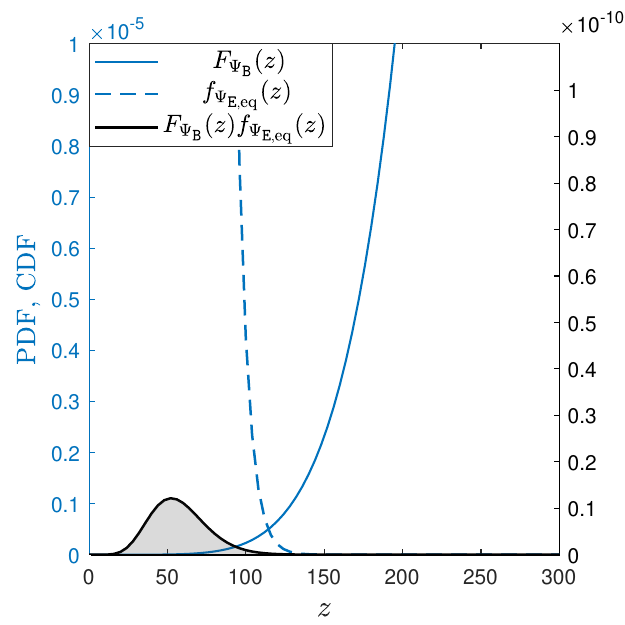}
      \caption{$L_{\Bob{}}=3$, Channel \Eve{}: Rician, $K_{\Eve{}}=4$.}
      \label{fig:MFTR_RicK4_Lb3Le1_PDFCDFs}
    \end{subfigure}
    \vspace{0.2cm}

    \begin{subfigure}[c]{0.3\textwidth}
      \includegraphics[width=4.7cm, height=5.2cm, left]{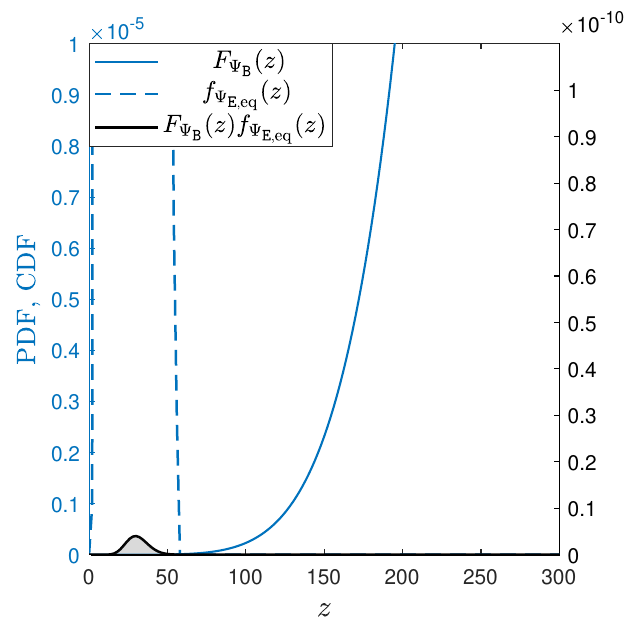}
      \caption{$L_{\Bob{}}=3$, Channel \Eve{}: Rician, $K_{\Eve{}}=25$.}
      \label{fig:MFTR_RicK25_Lb3Le1_PDFCDFs}
    \end{subfigure}
    
  \end{tabular}
  \vspace{-0.3cm}
  \caption{$F_{\Bob{}}(z)$ and $f_{\Eve{}\text{eq}}(z)$ and their product accordingly to \eqref{eq:SOP}, parameterized by $L_{\Bob{}}$  and the fading channel experiencing by \Eve{}, considering $m_{\Bob{}}=3, \mu_{\Bob{}} = 2, \sigma_{\Bob{}}^{2}=0.5, K_{\Bob{}}= 2.5, \Delta_{\Bob{}}=0.8$ in $F_{\Psi_{\Bob{}}}(z)$, $L_{\Eve{}}=1$, and the fading parameters in Table \ref{tab:Fad_Par_Eve} in $f_{\Psi_{\Eve{},\text{eq}}}(z)$, and $\theta=1$. The left vertical axis is associated with $F_{\Bob{}}(z)$ and $f_{\Eve{}\text{eq}}(z)$, while the right vertical axis is associated with their product.}
  \label{fig:MFTR_EveChannel_PDFCDFs}
\end{figure*}

With \eqref{eq:SOP} and \eqref{eq:PDF_Psieq}, we can notice that the SOP is calculated as the area under the curve of the product of the CDF of $\Psi_{\Bob{}}$ and the PDF of $\Psi_{\Eve{},\text{eq}}$. These functions are displayed in Fig.~\ref{fig:MFTR_EveChannel_PDFCDFs} for different scenarios presented in Fig. \ref{fig:GSOPSim_fEsNob_parLb_Le1_EsNoe10_Theta1_MFTR_Chan}, along with the resulting product function $F_{\Psi_{\Bob{}}}(z) f_{\Psi_{\Eve{},\text{eq}}}(z)$. Additionally, the area under this curve has been shaded as it represents the SOP value. For better understanding and without loss of generality, we have selected scenarios where the fading affecting the $\Alice{}\rightarrow\Eve{}$ link is Rayleigh, Rician with $K_{\Eve{}} = 4$, and Rician with $K_{\Eve{}} = 25$. Additionally, the results are shown for $L_{\Bob{}}=1$ and $L_{\Bob{}}=3$. When comparing Figs. \ref{fig:MFTR_Ray_Lb1Le1_PDFCDFs}, \ref{fig:MFTR_RicK4_Lb1Le1_PDFCDFs}, and \ref{fig:MFTR_RicK25_Lb1Le1_PDFCDFs}, note that $F_{\Psi_{\Bob{}}}(z)$ curve does not change, as in all these cases, $L_{\Bob{}}=1$ is used; on the other hand, the $f_{\Psi_{\Eve{},\text{eq}}}(z)$ curve changes from figure to figure because the kind of fading affecting the $\Alice{}\rightarrow\Eve{}$ link is modified. When the fading distribution changes from Rayleigh to Rician, the tail of the product function decreases considerably, and the maximum value of its highest peak increases slightly. As a result, the area under this curve tends to decrease when the channel is Rician, and this reduction is greater as $K_{\Eve{}}$ increases. Nevertheless, in this specific instance with $L_{\Bob{}}=1$, notice that the numerical value of the shaded area remains 
similar across all three scenarios. Consequently, the GSOP shows a comparable trend in these cases, as depicted in the enlarged area of Fig. \ref{fig:GSOPSim_fEsNob_parLb_Le1_EsNoe10_Theta1_MFTR_Chan}.

In Figs. \ref{fig:MFTR_Ray_Lb3Le1_PDFCDFs}, \ref{fig:MFTR_RicK4_Lb3Le1_PDFCDFs}, and \ref{fig:MFTR_RicK25_Lb3Le1_PDFCDFs}, the $F_{\Psi_{\Bob{}}}(z)$ curve changes in relation to the figures analyzed in the previous paragraph; this is because this new CDF considers the case $L_{\Bob{}}=3$. On the other hand, the $f_{\Psi_{\Eve{},\text{eq}}}(z)$ curve is the same as presented in Figs. \ref{fig:MFTR_Ray_Lb1Le1_PDFCDFs}, \ref{fig:MFTR_RicK4_Lb1Le1_PDFCDFs}, and \ref{fig:MFTR_RicK25_Lb1Le1_PDFCDFs}, according to the corresponding fading channel affecting the $\Alice{}\rightarrow\Eve{}$ link\footnote{The $f_{\Psi_{\Eve{},\text{eq}}}(z)$ curve appears different across Figs. \ref{fig:MFTR_Ray_Lb1Le1_PDFCDFs}–\ref{fig:MFTR_RicK25_Lb1Le1_PDFCDFs} and Figs. \ref{fig:MFTR_Ray_Lb3Le1_PDFCDFs}–\ref{fig:MFTR_RicK25_Lb3Le1_PDFCDFs} due to varying axis scales used to enhance result visualization.}. 
In Figs. \ref{fig:MFTR_Ray_Lb3Le1_PDFCDFs}, \ref{fig:MFTR_RicK4_Lb3Le1_PDFCDFs}, and \ref{fig:MFTR_RicK25_Lb3Le1_PDFCDFs}, note that the resulting product function changes considerably when the type of fading affecting the $\Alice{}\rightarrow\Eve{}$ link is modified. Thus, when comparing, for example, Fig. \ref{fig:MFTR_Ray_Lb3Le1_PDFCDFs} with Fig. \ref{fig:MFTR_RicK25_Lb3Le1_PDFCDFs}, a significant reduction in the area under this curve is evident in this last figure. In the case of Fig. \ref{fig:MFTR_Ray_Lb3Le1_PDFCDFs}, the highest peak of the PDF is in a region where the CDF is increasing.
On the other hand, in Fig. \ref{fig:MFTR_RicK25_Lb3Le1_PDFCDFs}, the highest peak of the PDF is in a region where the CDF has minimal values. As a result, in this latter case, the area under the product function is relatively small compared to that of Fig. \ref{fig:MFTR_Ray_Lb3Le1_PDFCDFs}, implying a much lower GSOP when the fading channel affecting the link $\Alice{}\rightarrow\Eve{}$ is Rician with $\mathcal{K}_{\Eve{}}=25$. This is also evident in Fig. \ref{fig:MFTR_EveChannel_PDFCDFs} when comparing Rayleigh and Rician scenarios for $L_{\Bob{}}=3$.

The previous results show that PLS depends on the type of fading channel affecting both the $\Alice{}\rightarrow\Bob{}$ and $\Alice{}\rightarrow\Eve{}$ links, or equivalently, on the behavior of $F_{\Psi_{\Bob{}}}(z)$ and $f_{\Psi_{\Eve{},\text{eq}}}(z)$, respectively. It is also interesting to note that $L_{\Bob{}}$ has a significant impact on system security. Specifically, it influences the GSOP diversity order and modifies the coding gain in accordance with the fading characteristics of the $\Alice{}\rightarrow\Eve{}$ link.

\subsection{AFE and AILR}

In the following, the exact AFE and AILR are calculated using \eqref{eq:AFE_Exact} and \eqref{eq:AILR_Exact}, respectively, while the approximate AFE and AILR are plotted using \eqref{eq:AFE_Approx} and \eqref{eq:AILR_Approx}.

Figs. \ref{fig:AFE_AILR_fEsNob_parLb_Le2_EsNoe_8dB} and \ref{fig:AFE_AILR_fEsNob_parLe_Lb2_EsNoe_8dB} depict AFE and AILR versus $\overline{\gamma_{\Bob{}}}$. The former sets $L_{\Eve{}}=2$ while varying $L_{\Bob{}}$, whereas the latter fixes $L_{\Bob{}}=2$ and varies $L_{\Eve{}}$.
Besides, in both figures is assumed that $\overline{\gamma_{\Eve{}}} = 8$ dB, $m_{\Bob{}}=4, \mu_{\Bob{}} = 2, \sigma_{\Bob{}}^{2}=0.5, K_{\Bob{}}= 2.5, \Delta_{\Bob{}}=0.8, m_{\Eve{}}=3, \mu_{\Eve{}} = 2, \sigma_{\Eve{}}^{2}=0.5, K_{\Eve{}}= 6.5$, and $\Delta_{\Eve{}}=0.9$. Notice that our exact expressions align with the simulation results, while the approximations exhibit high accuracy in capturing the simulation outcomes. Observe that with the increase in $\overline{\gamma_{\Bob{}}}$, the AFE also rises. Since the AFE serves as an asymptotic lower bound for \Eve{}'s decoding error probability, as $\overline{\gamma_{\Bob{}}}$ increases, the error probability for \Eve{} also increases.
On the other hand, as $\overline{\gamma_{\Bob{}}}$ increases, notice that the AILR decreases. This indicates that as \Bob{}'s SNR rises, the rate of information leakage to \Eve{} diminishes.
Moreover, note from Fig. \ref{fig:AFE_AILR_fEsNob_parLb_Le2_EsNoe_8dB} that, as $L_{\Bob{}}$ increases, the AFE increases while the AILR diminishes. This indicates that the secrecy performance improves as the diversity order at \Bob{} increases. 
In contrast, in Fig. \ref{fig:AFE_AILR_fEsNob_parLe_Lb2_EsNoe_8dB}, as $L_{\Eve{}}$ increases, the AFE diminishes while the AILR increases. This indicates that enhancing the diversity gain at \Eve{} lowers its decoding error probability while increasing the information leakage rate, as expected.

\begin{figure}[t]
    \centerline{\includegraphics
    [width=8.3cm, height=5.3cm]
    {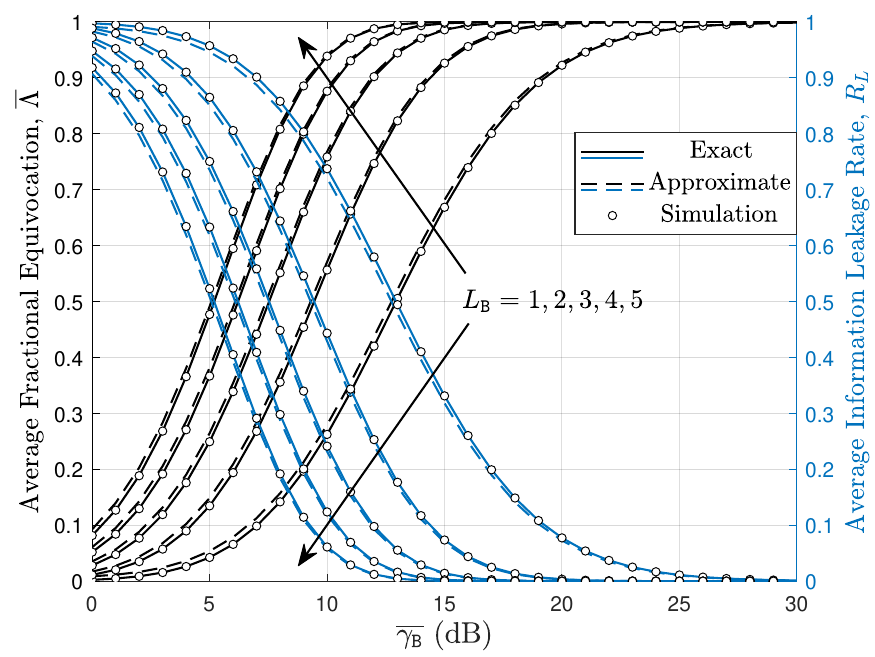}}
    \caption{AFE and AILR as a function of $\overline{\gamma_{\Bob{}}}$, parameterized by $L_{\Bob{}}$ considering $L_{\Eve{}}= 2, \overline{\gamma_{\Eve{}}} = 8$ dB, $m_{\Bob{}}=4, \mu_{\Bob{}} = 2, \sigma_{\Bob{}}^{2}=0.5, K_{\Bob{}}= 2.5, \Delta_{\Bob{}}=0.8, m_{\Eve{}}=3, \mu_{\Eve{}} = 2, \sigma_{\Eve{}}^{2}=0.5, K_{\Eve{}}= 6.5$, and $\Delta_{\Eve{}}=0.9$.}
    \label{fig:AFE_AILR_fEsNob_parLb_Le2_EsNoe_8dB}
    \vspace{-3mm}
\end{figure}

\begin{figure}[t]
    \centerline{\includegraphics
    [width=8.3cm, height=5.3cm]
    {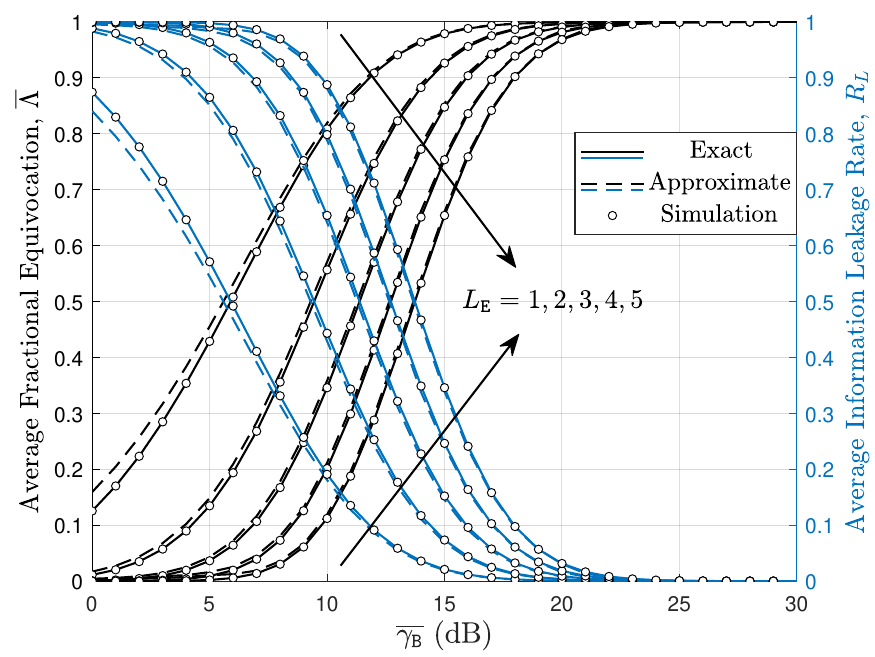}}
    \caption{AFE and AILR as a function of $\overline{\gamma_{\Bob{}}}$, parameterized by $L_{\Eve{}}$ considering $L_{\Bob{}}= 2, \overline{\gamma_{\Eve{}}} = 8$ dB, $m_{\Bob{}}=4, \mu_{\Bob{}} = 2, \sigma_{\Bob{}}^{2}=0.5, K_{\Bob{}}= 2.5, \Delta_{\Bob{}}=0.8, m_{\Eve{}}=3, \mu_{\Eve{}} = 2, \sigma_{\Eve{}}^{2}=0.5, K_{\Eve{}}= 6.5$, and $\Delta_{\Eve{}}=0.9$.}
    \label{fig:AFE_AILR_fEsNob_parLe_Lb2_EsNoe_8dB}
    \vspace{-3mm}
\end{figure}

Fig. \ref{fig:AFE_AILR_fRs_parme_EsNoe_LbLe2_EsNob20dB} shows the AFE and AILR versus $R_{s}$, parameterized by $m_{\Eve{}}$ and $\overline{\gamma_{\Eve{}}}$, considering $L_{\Bob{}}= 2, \overline{\gamma_{\Bob{}}} = 20$ dB, $m_{\Bob{}}=4, \mu_{\Bob{}} = 2, \sigma_{\Bob{}}^{2}=0.5, K_{\Bob{}}= 2.5, \Delta_{\Bob{}}=0.8, \mu_{\Eve{}} = 2, \sigma_{\Eve{}}^{2}=0.5, K_{\Eve{}}= 6.5$, and $\Delta_{\Eve{}}=0.9$. 
The results show that as $R_s$ increases, the AFE gradually decreases. This behavior reflects the fact that the FE---an asymptotic lower bound on \Eve{}’s decoding error probability---is less likely to approach 1 under more stringent secrecy constraints. As a result, its expected value, the AFE, diminishes, indicating that \Eve{}'s capacity to extract confidential information increases as the system targets higher secrecy rates. This illustrates a trade-off: while a higher $R_s$ represents a stronger secrecy requirement, it also reduces $\Eve{}$'s decoding error probability. 

In contrast, we can observe from Fig. \ref{fig:AFE_AILR_fRs_parme_EsNoe_LbLe2_EsNob20dB} that the AILR increases monotonically with $R_s$. As defined in \eqref{eq:AILR}, the leakage rate grows as a result of both the increase in $R_s$ and the decrease in AFE. This implies that although a higher $R_s$ enforces stricter secrecy conditions, it also raises the amount of information potentially leaked to \Eve{} over time. Hence, increasing $R_s$ does not inherently improve the overall secrecy performance; rather, it must be balanced carefully with the system's fading characteristics and \Eve{}'s channel quality.
Note also that, with an increase in $\overline{\gamma_{\Eve{}}}$, the AFE decreases while the AILR increases for a given $R_s$. This aligns with expectations, as a higher SNR at \Eve{} reduces the lower bound on its decoding error probability while increasing the information leakage rate.
Finally, the influence of $m_{\Eve{}}$ is evident, where an increase from 1 to 20 signifies a less severe shadowing impact on the specular components of the $\Alice{}\rightarrow\Eve{}$ link. This improvement in the propagation conditions for \Eve{} results in a lower AFE and a higher AILR as $R_s$ increases.

\begin{figure}[t]
    \centerline{\includegraphics[width=\linewidth]
    {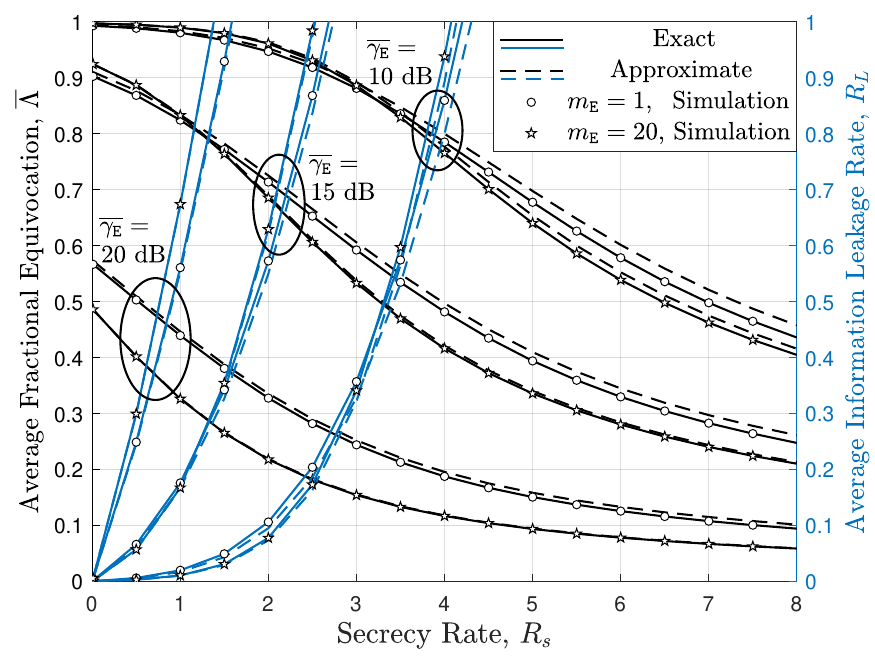}}
    \caption{AFE and AILR as a function of $R_{s}$, parameterized by $m_{\Eve{}}$ and $\overline{\gamma_{\Eve{}}}$ considering $L_{\Bob{}}= 2, \overline{\gamma_{\Bob{}}} = 20$ dB, $m_{\Bob{}}=4, \mu_{\Bob{}} = 2, \sigma_{\Bob{}}^{2}=0.5, K_{\Bob{}}= 2.5, \Delta_{\Bob{}}=0.8, \mu_{\Eve{}} = 2, \sigma_{\Eve{}}^{2}=0.5, K_{\Eve{}}= 6.5$, and $\Delta_{\Eve{}}=0.9$.}
    \label{fig:AFE_AILR_fRs_parme_EsNoe_LbLe2_EsNob20dB}
\end{figure}

\subsection{Truncation Error}

Fig. \ref{fig:Trunc_fT_parLb_Le_EsNob_30dB_EsNoe_8dB} shows the truncation error for the GSOP, $\mathcal{R}_{T}$, as a function of the truncation limit, $T$, parameterized by $L_{\Bob{}}$ and $L_{\Eve{}}$, considering $\overline{\gamma_{\Bob{}}} = 30$ dB, $\overline{\gamma_{\Bob{}}} = 8$ dB, and the same channel parameters used in Fig. \ref{fig:GSOP_fEsNob_parLb_Le_EsNoe_8dB_Theta05}. Specifically, these results were obtained using~\eqref{eq:ratio_bound}. The curves are plotted only after the stabilization of the ratio $r_j$ defined in \eqref{eq:rj}. In this context, stabilization means that the sequence of ratios, which may initially fluctuate due to oscillations in the recursive coefficients $\varphi_{j,\Eve{}}$ and the confluent hypergeometric term, reaches a regime where the ratio decreases with $j$. Beyond this point, $r_j$ can be regarded as representative of the subsequent terms, so the ratio bound in~\eqref{eq:ratio_bound} provides a reliable estimate of the truncation error.

Fig.~\ref{fig:Trunc_fT_parLb_Le_EsNob_30dB_EsNoe_8dB} shows that as the truncation limit $T$ increases, the truncation error $\mathcal{R}_{T}$ decreases, which is an expected behavior. An interesting observation is that when $L_{\Bob{}}$ increases, the truncation error becomes smaller. This is consistent with the fact that a larger number of antennas at \Bob{} reduces the GSOP, and thus the truncation error follows the same trend. Conversely, when $L_{\Eve{}}$ increases, the truncation error grows, since more antennas at \Eve{} lead to a higher GSOP, and the truncation error again reflects this behavior. 
To illustrate, consider two representative cases. First, when $L_{\Bob{}}=1$, $L_{\Eve{}}=2$, and $T=50$, the truncation error in Fig.~\ref{fig:Trunc_fT_parLb_Le_EsNob_30dB_EsNoe_8dB} is approximately $\mathcal{R}_{T}\approx 3\times 10^{-5}$. From Fig.~\ref{fig:GSOP_fEsNob_parLb_Le_EsNoe_8dB_Theta05}, the corresponding GSOP for $\overline{\gamma_{\Bob{}}} = 30$ dB is about $1.1 \times 10^{-3}$. This demonstrates that truncating the infinite series with $T=50$ terms makes the truncation error practically negligible. As a second example, when $L_{\Bob{}}=2$, $L_{\Eve{}}=4$, and $T=50$, we obtain $\mathcal{R}_{T}\approx 6\times 10^{-7}$, while the GSOP in Fig.~\ref{fig:GSOP_fEsNob_parLb_Le_EsNoe_8dB_Theta05} is approximately $4 \times 10^{-6}$ for $\overline{\gamma_{\Bob{}}} = 30$ dB. Again, the truncation error remains insignificant.

These results highlight that our analytical expressions can be truncated with a relatively small number of terms while still providing accurate and reliable evaluations of the GSOP. Furthermore, as both $L_{\Bob{}}$ and $L_{\Eve{}}$ increase, a larger number of terms may be required to achieve the same level of accuracy. Nevertheless, this number can be estimated from~\eqref{eq:ratio_bound}, ensuring that the truncation remains computationally feasible in practice.

\begin{figure}[t]
    \centerline{\includegraphics[width=\linewidth]
    {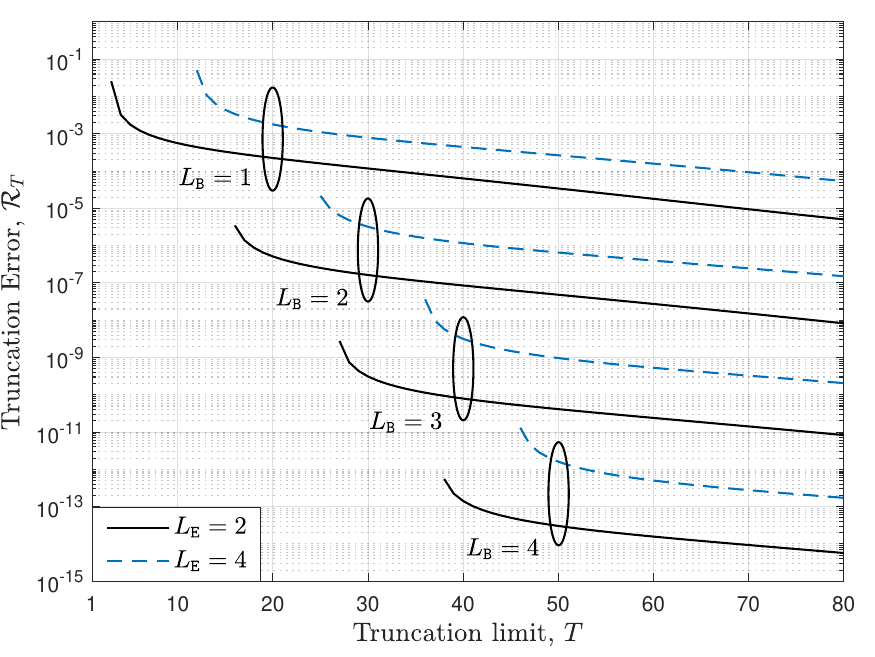}}
    \caption{Truncation error as a function of $T$, parameterized by $L_{\Bob{}}$ and $L_{\Eve{}}$ considering $m_{\Bob{}}=4, \mu_{\Bob{}} = 2, \sigma_{\Bob{}}^{2}=0.5, K_{\Bob{}}= 2.5, \Delta_{\Bob{}}=0.8$, $\overline{\gamma_{\Eve{}}} = 8$ dB, $m_{\Eve{}}=3, \mu_{\Eve{}} = 2, \sigma_{\Eve{}}^{2}=0.5, K_{\Eve{}}= 6.5$, $\Delta_{\Eve{}}=0.9$, and $\theta = 0.5$.}
    \label{fig:Trunc_fT_parLb_Le_EsNob_30dB_EsNoe_8dB}
    \vspace{-3mm}
\end{figure}


\section{Conclusions}
\label{sec: Conclusions}

This study advances the understanding of PLS in wireless systems operating over generalized MFTR fading channels by addressing critical gaps in existing research. We analyzed the partial secrecy regime, exploring a scenario involving a transmitter, a legitimate receiver, and an external eavesdropper.
Our system model assumed that both \Bob{} and \Eve{} are equipped with antenna arrays and apply maximal ratio combining (MRC). Additionally, the fading across the diversity branches at each receiver is modeled as i.i.d., while the fading affecting the $\Alice{}\rightarrow\Bob{}$ and $\Alice{}\rightarrow\Eve{}$ links is modeled as i.n.i.d.

We provided exact expressions for calculating the GSOP, AFE, and AILR, using the FE definition, offering insights into \Eve{}'s decoding capability and the amount of leaked information. 
Unlike previous studies based on simpler fading models, where diversity is typically considered only at \Eve{} and the number of nested infinite summations in SOP expressions scales with the number of diversity branches, our approach offers a significant improvement. Specifically, our derived expressions consist of only two infinite summations, regardless of the number of antennas at \Bob{} and \Eve{}. These summations can be efficiently truncated, enabling accurate and fast evaluation of secrecy performance.

Our study also offered closed-form approximate expressions in the high SNR regime for GSOP, AFE, and AILR, along with an asymptotic expression for GSOP. This revealed the impact of diversity and coding gain on secrecy performance. 
An important conclusion is that as the number of antennas at the legitimate receiver increases, the diversity order also increases, which is an expected outcome. However, as the number of antennas at the eavesdropper increases, a loss in SNR occurs, while the diversity order remains unchanged. Interestingly, the results indicate that increasing $L_{\Bob{}}$ leads to a significantly greater improvement in secrecy performance compared to the degradation observed when $L_{\Eve{}}$ increases. The advantage for the legitimate user becomes apparent, as it is able to better exploit diversity, even though both \Bob{} and \Eve{} employ MRC. 
Monte Carlo simulations validated the precision of our expressions across diverse scenarios. 

For future work, a scenario in which \Alice{} also employs multiple antennas, i.e., a multiple-input multiple-output (MIMO) system, could be considered. In particular, strategies such as transmit antenna selection (TAS) could be explored to optimize the trade-off between secrecy performance and system complexity under the MFTR fading. Our mathematical modeling can serve as a foundational framework for this purpose. However, extending such a scenario from our study is not trivial; it presents significant challenges that require careful consideration. 
Another promising direction involves incorporating correlation effects, either between the fading channels from \Alice{} to \Bob{} and from \Alice{} to \Eve{}, or across the antennas themselves. 
In this regard, a key open research challenge lies in the statistical characterization of correlated MFTR fading channels, which has not yet been addressed in the literature. 
Understanding the impact of such correlations on secrecy performance would be crucial for accurately modeling realistic propagation environments and strengthening the practical relevance of the proposed framework.

\section*{Acknowledgments}
We thank the Deep Learning Laboratory at Universidad San Francisco de Quito (USFQ) for providing computational resources and AI infrastructure that supported this research.

\begin{appendices}
\section{Proof of Proposition 2}
\label{app:Proof_CDF_Phi}
By definition, the CDF of $\Phi$ is obtained as
\begin{align}
    \label{eq:CDF_Phi}
    F_{\Phi}(z) = \text{P}(\Phi \leq z).
\end{align}
In light of \eqref{eq:Phi}, the expression in \eqref{eq:CDF_Phi} can be rewritten as
\begin{align}
    \label{eq:CDF_Phi_2}
    \nonumber F_{\Phi}(z) &= \text{P}\left(\Psi_{\Bob{}} \leq (\Psi_{\Eve{}}+1)z-1 \vphantom{2^2}\right) \\
    &= \int_{0}^{\infty}F_{\Psi_{\Bob{}}} \left((y+1)z-1 \vphantom{2^{2}}\right)f_{\Psi_{\Eve{}}}(y) \text{d}y.
\end{align}
Using \eqref{eq: sum PDF} and \eqref{eq: sum CDF}, considering that $\nu_{i,\Bob{}}\in \mathbb{Z}$, and the series representation of the lower incomplete gamma function \cite[eq. (8.69)]{Arfken2005}, \eqref{eq:CDF_Phi_2} can be now expressed as
\begin{align}
    \label{eq:CDF_Phi_3}
    \nonumber F_{\Phi} (z) =& \sum_{i=0}^{\infty} \varphi_{i,\Bob{}} \sum_{j=0}^{\infty} \varphi_{j,\Eve{}} \frac{\rho_{\Eve{}}^{-\nu_{j,\Eve{}}}}{\Gamma(\nu_{j,\Eve{}})} \\
    \nonumber &\times \left(\vphantom{\int_{0}^{\infty}}\right. \underbrace{ \int_{0}^{\infty} y^{\nu_{j,\Eve{}}-1}\exp\left(-\frac{y}{\rho_{\Eve{}}}\right)\text{d}y}_{\mathcal{I}_{1}(j)}\\
    \nonumber &- \sum_{a=0}^{\nu_{i,\Bob{}}-1}\frac{1}{a!}
    \underbrace{\int_{0}^{\infty} y^{\nu_{j,\Eve{}}-1} \left(\frac{(y+1)z-1}{\rho_{\Bob{}}}\right)^{a}}_{\mathcal{I}_{2}(i,j)} \\
    &   \underbrace{\times\exp\left(-\frac{y}{\rho_{\Eve{}}}\right) \exp\left(-\frac{(y+1)z-1}{\rho_{\Bob{}}}\right) \text{d}y}_{\mathcal{I}_{2}(i,j)}\left.\vphantom{\int_{0}^{\infty}} \right).
\end{align}
Using \cite[eq. (6.1.1)]{abramowitz72}, we obtain that $\mathcal{I}_{1}(j) = \rho_{\Eve{}}^{\nu_{j,\Eve{}}}\Gamma(\nu_{j,\Eve{}})$. Then, by solving $\mathcal{I}_{2}(i,j)$, it follows that
\begin{align}
    \label{eq:I2}
    \nonumber &\mathcal{I}_{2}(i,j) =\\
    \nonumber &\Gamma(a+\nu_{j,\Eve{}})\exp\left(- \frac{z-1}{\rho_{\Bob{}}}\right) \left(\frac{z}{\rho_{\Bob{}}}\right)^{a} \left(\frac{\rho_{\Bob{}}\rho_{\Eve{}}}{\rho_{\Bob{}} + z\rho_{\Eve{}}}\right)^{a+\nu_{j,\Eve{}}}\\
    &\times \,_1F_1 \left(-a,1-a-\nu_{j,\Eve{}},\frac{(z-1)(\rho_{\Bob{}} + z\rho_{\Eve{}})}{z\rho_{\Bob{}}\rho_{\Eve{}}}\right).
\end{align}
Finally, by employing the results of $\mathcal{I}_{1}(j)$ and $\mathcal{I}_{2}(i,j)$ in \eqref{eq:CDF_Phi_3}, and after some algebraic simplifications, $F_{\Phi}(z)$ can be rewritten as \eqref{eq:Phi_CDF}, which concludes the proof.

\section{Proof of Proposition 3} \label{app:Proof_GSOP}

From \eqref{eq:Frac_Eq_2} and \eqref{eq:GSOP}, the GSOP can be calculated as
\begin{align}
    \label{eq:GSOP_3}
    \nonumber \text{GSOP} =&\, \text{P}(0 < \theta | \Phi \leq 1) \text{P}(\Phi \leq 1) \,+ \\
    \nonumber &\, \text{P}\left(\frac{1}{R_{s}}\log_{2}\Phi < \theta \Big| 1 < \Phi < 2^{R_{s}}\right) \text{P}\left(1 < \Phi < 2^{R_{s}}\right)\\
    \nonumber  =&\, F_{\Phi}(1) + \text{P}\left(\frac{1}{R_{s}}\log_{2}\Phi < \theta \cap 1 < \Phi < 2^{R_{s}}\right) \\
    =&\, F_{\Phi}(1) + \text{P}\left(\Phi < 2^{\theta R_{s}} \cap 1 < \Phi < 2^{R_{s}}\right),
\end{align}
where we have used that $\text{P}(0 < \theta | \Phi \leq 1)=1$, given that $0<\theta\leq 1$, along with $\text{P}(A|B)=\text{P}(A \cap B)/\text{P}(B)$ and $\text{P}(Z<z)=F_{Z}(z)$. Now, by noticing that $1 < 2^{\theta R_{s}} \leq 2^{R_{s}}$, the GSOP can be reformulated as
\begin{align}
    \label{eq:GSOP_4}
    \text{GSOP} =\, F_{\Phi}(1) + \text{P}\left(1< \Phi \leq 2^{\theta R_{s}}\right).
\end{align}
Subsequently, employing the fact that $\text{P}(1 < \Phi \leq 2^{\theta R_{s}}) = F_{\Phi}(2^{\theta R_{s}}) - F_{\Phi}(1)$ in \eqref{eq:GSOP_4}, the GSOP can be written as
\begin{align}
    \label{eq:GSOP_5}
    \textnormal{GSOP} = F_{\Phi}\left(2^{\theta R_{s}}\right),
\end{align}
where $F_{\Phi}(z)$ is given by \eqref{eq:Phi_CDF} in Proposition \ref{prop:CDF_Phi}. Thus, the GSOP can eventually be written as in \eqref{eq:GSOP_2}. With this statement, the proof is concluded.

\section{Proof of Proposition 4}
\label{app:Proof_Exact_AFE}

From \eqref{eq:Frac_Eq_2} and \eqref{eq:AFE}, the AFE can be calculated as
\begin{align}
    \label{eq:AFE_2}
    \overline{\Lambda} = \frac{1}{R_{s}} \int_{1}^{2^{R_{s}}} \log_{2} (z) f_{\Phi}(z) \text{d}z + \int_{2^{R_{s}}}^{\infty} f_{\Phi}(z) \text{d}z
\end{align}
Now, employing the identity $\log_{2}(z) = \log(z)/\log(2)$, applying integration by parts to the first integral and using that the second integral is equal to $1 - F_{\Phi}\left(2^{R_{s}}\right)$, we obtain
\begin{align}
    \label{eq:AFE_3}
    \nonumber \overline{\Lambda} =& \frac{1}{\log \left(2^{R_{s}}\right)} \left( \log(z)F_{\Phi}(z)\Big|_{1}^{2^{R_{s}}} - \int_{1}^{2^{R_{s}}} \frac{1}{z}F_{\Phi}(z)\text{d}z \right)\\
    & + 1 - F_{\Phi}\left(2^{R_{s}}\right).
\end{align}
Then, using that $\log(1) = 0$ and after some simplifications, we finally obtain \eqref{eq:AFE_Exact}, which concludes the proof. 

\section{Proof of Proposition 5} \label{app:Proof_CDF_L}

The CDF of $\Phi_{\textnormal{A}} = \Psi_{\Bob{}}/\Psi_{\Eve{}}$ can be obtained as
\begin{align}
    \label{eq:CDF_Phi_L}
    \nonumber F_{\Phi_{\text{A}}}(z) &= \text{P}(\Phi_{\text{A}} \leq z)\\
    \nonumber & = \text{P}\left(\Psi_{\Bob{}} \leq z \Psi_{\Eve{}} \vphantom{2^2}\right) \\
    & = \int_{0}^{\infty}F_{\Psi_{\Bob{}}} (yz \vphantom{2^{2}})f_{\Psi_{\Eve{}}}(y) \text{d}y.
\end{align}
Now, by employing \eqref{eq: sum PDF} and \eqref{eq: sum CDF}, considering that $\nu_{i,\Bob{}}\in \mathbb{Z}$ and \cite[eq. (8.69)]{Arfken2005}, \eqref{eq:CDF_Phi_L} can be rewritten as
\begin{align}
    \label{eq:CDF_Phi_L_2}
    \nonumber F_{\Phi_{\text{A}}} (z) =& \sum_{i=0}^{\infty} \varphi_{i,\Bob{}} \sum_{j=0}^{\infty} \varphi_{j,\Eve{}} \frac{\rho_{\Eve{}}^{-\nu_{j,\Eve{}}}}{\Gamma(\nu_{j,\Eve{}})} \\
    \nonumber &\times \left(\vphantom{\int_{0}^{\infty}}\right. \underbrace{ \int_{0}^{\infty} y^{\nu_{j,\Eve{}}-1}\exp\left(-\frac{y}{\rho_{\Eve{}}}\right)\text{d}y}_{\mathcal{I}_{1}(j)} -  \sum_{a=0}^{\nu_{i,\Bob{}}-1}\frac{1}{a!}\\
    &\times \underbrace{\int_{0}^{\infty} y^{\nu_{j,\Eve{}}-1} \left(\frac{yz}{\rho_{\Bob{}}}\right)^{a}\exp\left(-\frac{y}{\rho_{\Eve{}}}-\frac{yz}{\rho_{\Bob{}}}\right)\text{d}y}_{\mathcal{I}_{3}(i,j)}\left.\vphantom{\int_{0}^{\infty}} \right).
\end{align}
Here, $\mathcal{I}_{1}(j)$ is the same integral defined in \eqref{eq:CDF_Phi_3}, thus being equal to $\mathcal{I}_{1}(j) = \rho_{\Eve{}}^{\nu_{j,\Eve{}}}\Gamma(\nu_{j,\Eve{}})$. Moreover, by solving $\mathcal{I}_{3}(i,j)$ we obtain
\begin{align}
    \label{eq:I3}
    \mathcal{I}_{3}(i,j) = \Gamma(a+\nu_{j,\Eve{}}) \left(\frac{z}{\rho_{\Bob{}}}\right)^{a} \left(\frac{\rho_{\Bob{}}\rho_{\Eve{}}}{\rho_{\Bob{}} + z\rho_{\Eve{}}}\right)^{a+\nu_{j,\Eve{}}}.
\end{align}
Using the results of $\mathcal{I}_{1}(j)$ and $\mathcal{I}_{3}(i,j)$ in \eqref{eq:CDF_Phi_L_2}, and after some simplifications, we can obtain $F_{\Phi_{\text{A}}}(z)$ as presented in \eqref{eq:CDF_Phi_Low}, which concludes the proof.

\section{Proof of Proposition 7} \label{app:GSOP_Assympt}

Upon substituting the exponential function $e^{-t}$ by its Mclaurin series expansion \cite[eq. (4.2.1)]{abramowitz72} in the integral definition of the lower incomplete gamma function \cite[eq. (6.5.2)]{abramowitz72} within the CDF of $\Psi_{\textnormal{\texttt{X}}}$, given by \eqref{eq: sum CDF}, we obtain
\begin{align}
    \label{eq:sum_CDF_New}
    \nonumber F_{\Psi_{\textbf{\texttt{X}}}} (z) &= \sum_{i=0}^{\infty} \frac{\varphi_{i,\textnormal{\texttt{X}}} }{\Gamma (\nu_{i,\textnormal{\texttt{X}}})}\int_{0}^{\frac{z}{\rho_{\texttt{X}}}} \sum_{a=0}^{\infty} \frac{(-1)^{a}}{a!} t^{\nu_{i,\texttt{X}}+a-1}\text{d}t \\
    &= \sum_{i=0}^{\infty} \frac{\varphi_{i,\textnormal{\texttt{X}}} }{\Gamma (\nu_{i,\textnormal{\texttt{X}}})} \sum_{a=0}^{\infty} \frac{(-1)^{a}}{a!(\nu_{i,\texttt{X}}+a)}\left(\frac{z}{\rho_{\texttt{X}}}\right)^{\nu_{i,\texttt{X}} + a}.
\end{align}
Building on the preceding outcome, we can establish an alternative expression for the CDF of $\Phi$, given by \eqref{eq:Phi}. Using \eqref{eq: sum PDF} and \eqref{eq:sum_CDF_New} into \eqref{eq:CDF_Phi_2}, and after some manipulations, we arrive at this alternative CDF expression as
\begin{align}
    \label{eq:sum_CDF_New_2}
    \nonumber &F_{\Phi} (z) = \\
    \nonumber &\sum_{i=0}^{\infty} \frac{\varphi_{i,\Bob{}}}{\Gamma (\nu_{i,\Bob{}})} \sum_{a=0}^{\infty} \frac{(-1)^{a}}{a!(\nu_{i,\Bob{}}+a)}\left(\frac{1}{\rho_{\Bob{}}}\right)^{\nu_{i,\Bob{}} + a}\sum_{j=0}^{\infty} \varphi_{j,\Eve{}} \frac{\rho_{\Eve{}}^{-\nu_{j,\Eve{}}}}{\Gamma(\nu_{j,\Eve{}})} \\ &\times \int_{0}^{\infty}  y^{\nu_{j,\Eve{}}-1} \left((y+1)z-1\right)^{\nu_{i,\Bob{}}+a} \exp\left(-\frac{y}{\rho_{\Eve{}}}\right)\text{d}y.
\end{align}
Upon solving the integral in \eqref{eq:sum_CDF_New_2} and simplifying the resultant expression, the alternative CDF of $\Phi$ can be reformulated~as
\begin{align}
    \label{eq:sum_CDF_New_3}
    \nonumber F_{\Phi} (z) = &\sum_{i=0}^{\infty} \frac{\varphi_{i,\Bob{}}}{\Gamma (\nu_{i,\Bob{}})} \sum_{a=0}^{\infty} \frac{(-1)^{a}}{a!(\nu_{i,\Bob{}}+a)}\left(\frac{z \rho_{\Eve{}}}{\rho_{\Bob{}}}\right)^{\nu_{i,\Bob{}} + a}\\
    \nonumber &\times \sum_{j=0}^{\infty} \varphi_{j,\Eve{}} \frac{\Gamma(\nu_{i,\Bob{}} + \nu_{j,\Eve{}} + a)}{\Gamma(\nu_{j,\Eve{}})}\\
    &\times \, _1F_1 \left(-a - \nu_{i,\Bob{}},1 - \nu_{i,\Bob{}} - \nu_{j,\Eve{}} - a,\frac{1-z}{\rho_{\Eve{}}}\right).
\end{align}
Using \eqref{eq:rho} and \eqref{eq:nu} to calculate $\rho_{\Bob{}}$ and $\nu_{i,\Bob{}}$, respectively, under the condition $\overline{\gamma_{\Bob{}}} \rightarrow \infty$, the most significant terms in the series of \eqref{eq:sum_CDF_New_3} are obtained for $i=0$ and $a=0$. Therefore, focusing on these predominant terms allows for an approximation of \eqref{eq:sum_CDF_New_3}, which is expressed as
\begin{align}
    \label{eq:sum_CDF_New_Approx}
    \nonumber F_{\Phi}^{\infty} (z) \approx & \frac{\varphi_{0,\Bob{}}}{\Gamma (\mu_{\Bob{}}L_{\Bob{}})\mu_{\Bob{}}L_{\Bob{}}} \left(\frac{z \rho_{\Eve{}}}{\rho_{\Bob{}}}\right)^{\mu_{\Bob{}}L_{\Bob{}}} \sum_{j=0}^{\infty} \varphi_{j,\Eve{}} \frac{\Gamma(\mu_{\Bob{}}L_{\Bob{}} + \nu_{j,\Eve{}})}{\Gamma(\nu_{j,\Eve{}})}\\
    &\times \, _1F_1 \left(- \mu_{\Bob{}}L_{\Bob{}},1 - \mu_{\Bob{}}L_{\Bob{}} - \nu_{j,\Eve{}} ,\frac{1-z}{\rho_{\Eve{}}}\right).
\end{align}
Lastly, leveraging the outcome of \eqref{eq:GSOP_2} as developed in Proposition \ref{prop:GSOP}, we can derive an approximation for the GSOP when $\overline{\gamma_{\Bob{}}}$ tends to infinity, expressed as $\textnormal{GSOP}^{\infty} = F_{\Phi}^{\infty} (2^{\theta R_{s}})$. This result leads to \eqref{eq:GSOP_Assymp}, which concludes the proof.

\section{Proof of Proposition 8} \label{app:AFE_Approx}

By substituting the CDF of $\Phi_{\textnormal{A}}$, as provided in \eqref{eq:CDF_Phi_Low}, into \eqref{eq:AFE_Exact} instead of the CDF of $\Phi$, an approximation for the AFE in the high SNR regime can be obtained as
\begin{align}
    \label{eq:AFE_Approx_2}
    \nonumber\overline{\Lambda}\vphantom{x}_{\textnormal{A}} =& 1 - \frac{1}{\log(2^{R_{s}})} \sum_{i=0}^{\infty}\varphi_{i,\Bob{}}\sum_{j=0}^{\infty} \varphi_{j,\Eve{}} \left(\vphantom{\int_{1}^{2^{R_{s}}}} \right.\underbrace{\int_{1}^{2^{R_{s}}}\frac{1}{z}\text{d}z}_{\mathcal{I}_{4}} - \frac{\rho_{\Bob{}}^{\nu_{j,\Eve{}}}}{\Gamma(\nu_{j,\Eve{}})} \\
    &\times  \sum_{a=0}^{\nu_{i,\Bob{}}-1} \frac{\Gamma(a + \nu_{j,\Eve{}})}{a!}\underbrace{\int_{1}^{2^{R_{s}}}\frac{1}{z}\frac{(z\rho_{\Eve{}})^{a}}{(\rho_{\Bob{}} + z\rho_{\Eve{}})^{a+\nu_{j,\Eve{}}}}\text{d}z}_{\mathcal{I}_{5}(i,j)} \left. \vphantom{\int_{1}^{2^{R_{s}}}}\right),
\end{align}
where $\mathcal{I}_{4} = \log(2^{R_{s}})$. Also, with the aid of \cite[eq. (8.391)]{Gradshteyn08} and after some manipulations, $\mathcal{I}_{5}(i,j)$ can be solved as
\begin{align}
    \label{eq:I5}
    \mathcal{I}_{5}(i,j) =& (-\rho_{\Bob{}})^{-\nu_{j,\Eve{}}}\left(\mathcal{B}\left(-\frac{\rho_{\Bob{}}}{\rho_{\Eve{}}},\nu_{j,\Eve{}},1-a-\nu_{j,\Eve{}} \right) \right.\\
    &\left. -\mathcal{B}\left(-\frac{\rho_{\Bob{}}}{2^{R_{s}}\rho_{\Eve{}}},\nu_{j,\Eve{}},1-a-\nu_{j,\Eve{}} \right)   \right).
\end{align}
Using the outcomes of $\mathcal{I}_{4}$ and $\mathcal{I}_{5}(i,j)$ in \eqref{eq:AFE_Approx_2}, and with some simplifications, we arrive at \eqref{eq:AFE_Approx}, concluding the proof.

\section{Truncation Error Estimation} \label{app:Ratio_Bound}

Consider the truncated series in~\eqref{eq:GSOP_Assymp_trunc}, indexed by $j$, for $j=0$ to $T$. The remainder after truncating at index $T$ is given by the sum of all subsequent terms. In other words, the truncation error of the series is defined as
\begin{align}
\zeta_T = \sum_{j=T+1}^{\infty} a_{j},
\end{align}
where $a_{j}$ is given in \eqref{eq:a_j}.
To obtain an upper bound for the absolute value of the remainder, we can use the triangle inequality as follows
\begin{align}
|\zeta_T| = \left|\sum_{j=T+1}^{\infty} a_{j}\right| \le \sum_{j=T+1}^{\infty} |a_{j}|.
\end{align}
Now, we define the ratio of successive terms at the truncation point as $r_{T+1} = |a_{T+2}| / |a_{T+1}|$. If we assume that this ratio and all subsequent ratios are less than or equal to\footnote{The convergence of the series does not guarantee that $|a_{j+1}|/|a_j| \le r_{T+1}$ for all $j > T$. In practice, the ratios may oscillate before stabilizing, so this condition must be verified numerically before applying the ratio bound.} $r_{T+1}$, we can rewrite the subsequent terms of the series
\begin{align}
|a_{T+2}| &\le r_{T+1} |a_{T+1}|, \\
|a_{T+3}| &\le r_{T+2} |a_{T+2}| \le r_{T+1}^2 |a_{T+1}|, \\
|a_{T+n}| &\le r_{T+1}^{n-1} |a_{T+1}|.
\end{align}
By substituting these inequalities into the upper bound of the remainder, we obtain a geometric series that bounds it as
\begin{align}
|\zeta_T| \le |a_{T+1}| + r_{T+1} |a_{T+1}| + r_{T+1}^2 |a_{T+1}| + \dots
\end{align}
Factoring out $|a_{T+1}|$, we obtain
\begin{align}
|\zeta_T| \le |a_{T+1}| (1 + r_{T+1} + r_{T+1}^2 + \dots)
\end{align}
Since the series is convergent, $r_{T+1} < 1$, and the sum of the infinite geometric series in parentheses is $\frac{1}{1-r_{T+1}}$ \cite[eq. (0.231.1)]{Gradshteyn08}. By substituting this sum, we obtain 
\begin{align}
|\zeta_T| \le \frac{|a_{T+1}|}{1 - r_{T+1}}.
\end{align}
Finally, by including the terms in~\eqref{eq:GSOP_Assymp_trunc} that are independent of the summation index $j$, an upper bound for the truncation error of the GSOP can obtained as
\begin{align}\label{eq:ratio_bound_1}
|\mathcal{R}_T| & \le \frac{\varphi_{0,\Bob{}}\left(\rho_{\Eve{}}\mu_{\Bob{}}(K_{\Bob{}}+1)2^{\theta R_{s}}\right)^{\mu_{\Bob{}}L_{\Bob{}}}}{\Gamma(\mu_{\Bob{}}L_{\Bob{}})\mu_{\Bob{}}L_{\Bob{}}} |\zeta_{T}|\left(\frac{1}{\,\overline{\gamma_{\Bob{}}}\,}\right)^{\mu_{\Bob{}}L_{\Bob{}}},
\end{align}
which can be rewritten as~\eqref{eq:ratio_bound}.

\end{appendices}

\bibliographystyle{IEEEtran}
\bibliography{Bibliographyr}

\end{document}